\documentclass[a4paper,USenglish,cleveref,autoref,thm-restate]{lipics-v2021}

\usepackage{hyperref}

\usepackage{todonotes}

\graphicspath{{./figures/}}

\def\computationproblem#1#2#3{
  \begin{center}
  \begin{tabular}{rp{0.8\textwidth}}
  {\sc Problem:\enspace}&#1\\
  {\sc Input:\enspace}&#2\\
  {\sc Question:\enspace}&#3\\
  \end{tabular}
  \end{center}
}

\def\dar{degree adjusting reduction}
\def\DAR{Degree adjusting reduction}
\def\bg{block graph}

\def\ubg{uniblock graph}
\def\interbg{interblock graph}

\newtheorem*{conjecture2}{Conjecture}

\title{Computational Complexity of Covering Colored Mixed Multigraphs with Simple Degree Partitions
\thanks{The conference version of this paper appeared in the proceedings of WG'23~\cite{n:BFJKS23-WG}}} 

\titlerunning{Computational Complexity of Covering Colored Mixed Multigraphs...}

\author{Jan Bok}{Department of Algebra, Faculty of Mathematics and Physics, Charles University, Prague, Czech Republic\\ Université Clermont Auvergne, CNRS, Mines de Saint-Étienne, Clermont-Auvergne-INP,
LIMOS, 63000, Clermont-Ferrand, France}{jan.bok@matfyz.cuni.cz}{https://orcid.org/0000-0002-7973-1361}{Supported by research grant GAČR 20-15576S of the Czech Science Foundation, by the International Research Center "Innovation Transportation and Production Systems" of
the I-SITE CAP 20-25 and the ANR project GRALMECO (ANR-21-CE48-0004), and by the
European Union (ERC, POCOCOP, 101071674). Views and opinions expressed are however those of the author(s) only and do
not necessarily reflect those of the European Union or the European Research Council Executive Agency. Neither the European
Union nor the granting authority can be held responsible for them.}
\author{Jiří Fiala}{Department of Applied Mathematics, Faculty of Mathematics and Physics, Charles University, Prague, Czech Republic}{fiala@kam.mff.cuni.cz}{https://orcid.org/0000-0002-8108-567X}{Supported by research grant GAČR 20-15576S of the Czech Science Foundation.}
\author{Nikola Jedličková}{Department of Applied Mathematics, Faculty of Mathematics and Physics, Charles University, Prague, Czech Republic}{jedlickova@kam.mff.cuni.cz}{https://orcid.org/0000-0001-9518-6386}{Supported by research grant GAČR 20-15576S of the Czech Science Foundation and by SVV--2020--260578.}
\author{Jan Kratochvíl}{Department of Applied Mathematics, Faculty of Mathematics and Physics, Charles University, Prague, Czech Republic}{honza@kam.mff.cuni.cz}{https://orcid.org/0000-0002-2620-6133}{Supported by research grant GAČR 20-15576S of the Czech Science Foundation.}
\author{Michaela Seifrtová}{Department of Applied Mathematics, Faculty of Mathematics and Physics, Charles University, Prague, Czech Republic}{michaela.seifrtova@mff.cuni.cz}{https://orcid.org/000-0003-0050-480X}{Supported by research grant GAČR 20-15576S of the Czech Science Foundation.}

\authorrunning{J. Bok, J. Fiala, N. Jedličková, J. Kratochvíl, M. Seifrtová}

\Copyright{Jan Bok and Jiří Fiala and Nikola Jedličková and Jan Kratochvíl and Micheala Seifrtová} 



\nolinenumbers %uncomment to disable line numbering

\hideLIPIcs  %uncomment to remove references to LIPIcs series (logo, DOI, ...), e.g. when preparing a pre-final version to be uploaded to arXiv or another public repository

\EventEditors{John Q. Open and Joan R. Access}
\EventNoEds{2}
\EventLongTitle{42nd Conference on Very Important Topics (CVIT 2016)}
\EventShortTitle{CVIT 2016}
\EventAcronym{CVIT}
\EventYear{2016}
\EventDate{December 24--27, 2016}
\EventLocation{Little Whinging, United Kingdom}
\EventLogo{}
\SeriesVolume{42}
\ArticleNo{23}

\begin{document}

\maketitle

\begin{abstract}
The notion of graph covers (also referred to as locally bijective homomorphisms) plays an important role in topological graph theory and has found its computer science applications in models of local computation. For a fixed target graph $H$, the {\sc $H$-Cover} problem asks if an input graph $G$ allows a graph covering projection onto $H$. Despite the fact that the quest for characterizing the computational complexity of {\sc $H$-Cover} had been started more than 30 years ago, only a handful of general results have been known so far.  

In this paper, we present a complete characterization of the computational complexity of covering coloured graphs for the case that every equivalence class in the degree partition of the target graph has at most two vertices. We prove this result in a very general form. Following the lines of current development of topological graph theory, we study graphs in the most relaxed sense of the definition. In particular, we consider graphs that are mixed (they may have both directed and undirected edges), may have multiple edges, loops, and semi-edges.  We show that a strong P/NP-complete dichotomy holds true in the sense that for each such fixed target graph $H$, the {\sc $H$-Cover} problem is either polynomial-time solvable for arbitrary inputs, or NP-complete even for simple input graphs. 
\end{abstract}

\section{Introduction}\label{sec:Intro}

The notion of {\em graph covers} stems from topology and is viewed as a discretization of the notion of covers of topological spaces. Apart from being used in combinatorics as a tool for constructing large highly symmetric graphs~\cite{n:Biggs74,n:Biggs81,n:Biggs82,n:Biggs84}, this notion has found computer science applications in the theory of local computation~\cite{n:Angluin80,n:Chalopin05,n:ChMZ06,n:ChalopinP11,n:CM94,n:LMZ93}.
{In this paper we aim to contribute to the kaleidoscope of results about computational complexity of graph covers. We first briefly comment on the known results and show where our main result is placed among them. The formal definitions of graphs under consideration (Definition~\ref{def:d-graph}) and of graph covering projections (Definitions~\ref{def:graph-cover} and~\ref{def:disconnected-cover}) are presented in Section~\ref{sec:Prelim}, as well as the detailed definition of the so called {\dar} (Definition~\ref{def:reduction}), the concept of the degree partition of a graph (Proposition~\ref{prop:degpart}) and identification of several special graphs which play the key role in our characterization in Theorem~\ref{thm:main} (Definition~\ref{def:key-graphs}).} 

Despite the efforts and attention that graph covers received in the computer science community, their computational complexity is still far from being fully understood. Bodlaender~\cite{n:Bodlaender89} proved that deciding if one graph covers another one is an NP-complete problem, if both graphs are part of the input. Abello, Fellows, and Stillwell~\cite{n:AFS91} considered the variant when the target graph, say $H$, is fixed, i.e., a parameter of the problem.

\computationproblem{{\sc $H$-Cover}}{A graph $G$.}{Does $G$ cover $H$?}

They showed examples of graphs $H$ for which the problem is polynomial-time solvable as well as examples for which it is NP-complete, but most importantly, they were the first to formulate the goal of a complete characterization of the computational complexity of the {\sc $H$-Cover} problem, depending on the target graph $H$. Some of the explicit questions of~\cite{n:AFS91} were answered by Kratochvíl, Proskurowski, and Telle~\cite{n:KPT94,n:KPT97}, some of the NP-hardness results have been strengthened to planar input graphs by Bílka, Jir\'asek, Klav\'{\i}k, Tancer, and Volec~\cite{n:BilkaJKTV11}. A connection to a generalization of the Frequency Assignment Problem has been identified through partial covers by Fiala and Kratochv\'{\i}l in~\cite{n:FK01}, with further results proven in~\cite{n:BLT11,n:FKP08}. The computationally even more sophisticated problem of {\em regular covers} has been treated in~\cite{n:FKKN14}. 
In a recent paper~\cite{n:BFHJK21-MFCS}, the authors initiated the study of the complexity of {\sc $H$-Cover} for graphs that allow multiple edges and loops, and also semi-edges. This is motivated by the recent development of topological graph theory where it has now become standard to consider this more general model of graphs~\cite{kwak2007graphs,n:MalnivcMP04,n:MalnicNS00,nedela_mednykh,n:NedelaS96}. The graphs with semi-edges were also introduced and used in mathematical physics, e.g.\ by Getzler and Karpanov~\cite{getzler1998modular}.
It should be pointed out right away that considering loops, multiple edges and directed edges was shown necessary already in~\cite{n:KPT97a}, where it is proven that in order to fully understand the computational complexity of {\sc $H$-Cover} for {\em simple} undirected graphs $H$ (i.e., undirected graphs without multiple edges, loops, and semi-edges), it is necessary and sufficient to understand the complexity of the problem for coloured mixed multigraphs of minimum degree greater than 2. All papers from that era restrict their attention to covers of connected graphs. Disconnected target graphs are carefully treated in detail only in~\cite{n:BFJKS24-DAM}, where it is argued that the right way to define covers of disconnected graphs is to request that the preimages of all vertices have the same size. Such covers are called {\em equitable covers} in~\cite{n:BFJKS24-DAM}, and in this paper we adopt this view and require graph covers to be equitable in case of covering disconnected graphs.    
Apart from several isolated results (which also include a complete characterization of the complexity of {\sc $H$-Cover} for  connected simple undirected graphs $H$ with at most 6 vertices~\cite{n:KPT94}, and proving that several cases of {\sc $H$-Cover}, including $H=K_4$, are NP-complete for planar input graphs~\cite{n:BilkaJKTV11}), 
the following general results have been known about the complexity of {\sc $H$-Cover} for infinite classes of graphs:

\begin{enumerate}
\item\label{result:PSimple} A polynomial-time algorithm for {\sc $H$-Cover} for connected simple undirected graphs $H$ that have at most two vertices in every equivalence class of their degree partitions~\cite{n:KPT94}.

\item NP-completeness of {\sc $H$-Cover} for regular simple undirected graphs $H$ of valency at least three~\cite{n:FK08,n:KPT97}.

\item Complete characterization of the complexity of {\sc $H$-Cover} for undirected (multi)graphs $H$ (without semi-edges) on at most three vertices~\cite{n:KTT16}.

\item\label{result:Mix} Complete characterization of the complexity of {\sc $H$-Cover} for coloured mixed (multi)graphs $H$ on at most two vertices~\cite{n:KPT97a} (for graphs without semi-edges) and~\cite{n:BFHJK21-MFCS} (with semi-edges allowed). 
\end{enumerate}

It turns out that so far all the known NP-hard instances of {\sc $H$-Cover} remain NP-hard for simple input graphs. This has led the authors of~\cite{n:BFJKR24-Algorithmica} to formulate the following conjecture.

\begin{conjecture2}[Strong Dichotomy Conjecture for Graph Covers~\cite{n:BFJKR24-Algorithmica}]
For every graph $H$, the {\sc $H$-Cover} problem  is either polynomial-time solvable for arbitrary input graphs, or it is NP-complete for simple  graphs as input.
\end{conjecture2}

Here our primary objective is to revisit the result of~\cite{n:KPT94} and generalize it in line with the current trends and development of topological graph theory to graphs with loops, semi-edges, and multiple edges.
The main result of our paper is a complete characterization of the computational complexity of {\sc $H$-Cover} for graphs $H$, each of whose equivalence classes of the degree partition has at most 2 vertices. This provides a common generalization of the aforementioned results \ref{result:PSimple} and \ref{result:Mix}.

\begin{theorem}\label{thm:newmain}
The {\sc $H$-Cover} problem satisfies Strong Dichotomy for graphs $H$ such that each equivalence class of the degree partition has at most 2 vertices --- it is either polynomial-time solvable for general graphs on input, or it is NP-complete for simple input graphs.
\end{theorem}

The actual characterization is rather technical and it follows from Theorem~\ref{thm:main} in Section~\ref{sec:Prelim}, presented after the formal definitions of all the notions and special graphs that are needed for it.
The characterization extends well beyond the motivating results presented in~\cite{n:KPT94,n:KPT97a}. The key novel contributions are as follows:

\begin{itemize}
\item For simple graphs $H$, the {\sc $H$-Cover} problem is always polynomial-time solvable (if $H$ has all equivalence classes of size at most 2), while for general graphs, already graphs with 2 vertices may define NP-complete cases. Indeed, when semi-edges are allowed, some graphs on a single vertex may also yield NP-completeness.

\item For simple graphs $H$, the polynomial time algorithm is based on {\sc 2-Sat}, while in case of general graphs, our polynomial time algorithm is a blend of {\sc 2-Sat} and {\sc Perfect Matching} algorithms; this is surprising, as these two approaches are known to be incompatible in some other situations.

\item The NP-complete cases are proven for simple input graphs, which is in line with the Strong Dichotomy Conjecture~\cite{n:BFJKR24-Algorithmica}. This is in contrast to many previous results which allowed multiple edges and loops in the input graphs. 

\end{itemize} 

\section{Preliminaries}\label{sec:Prelim}

\subsection{Definitions}

Throughout the paper we consider the most general notion of a {\em graph} which allows multiple edges, loops, directed edges, and also semi-edges, and whose elements --- both edges and vertices --- are coloured. A semi-edge is a pendant edge, incident with just one vertex (and adding just $1$ to the degree of this vertex, unlike the loop, which adds $2$ to the degree). In figures, semi-edges are depicted as lines with one loose end, the other one being the vertex incident to the semi-edge.  To avoid any possible confusion, we present a formal definition.

\begin{definition}\label{def:d-graph} 
  A {\em graph} is a quadruple $G=(V,\Lambda,\iota,c)$, where $V$ is a (finite) set of {\em vertices},
  $\Lambda=\overline{E}\cup\overrightarrow{E}\cup\overline{L}\cup\overrightarrow{L}\cup S$ is the set of edges of $G$, $\iota:\Lambda\to {V\choose 2}\cup (V\times V)\cup V$ is the incidence mapping of edges, and $c:V\cup E\to C$ is a colouring of the vertices and edges. There are several types of edges:
  \begin{itemize}
  \item the edges of $\overline{E}$ are called {\em normal undirected edges} and they satisfy $\iota(e)\in{V\choose 2}$, 
  \item the edges of $\overrightarrow{E}$ are called {\em normal directed edges} and they satisfy $\iota(e)\in (V\times V)\setminus\{(u,u):u\in V\}$, 
  \item the edges of $\overline{L}$ ($\overrightarrow{L}$) are called {\em undirected (directed}, respectively) {\em loops} and we have $\iota(e)\in V$ in both cases, and 
  \item the edges of $S$ are called {\em semi-edges} and again $\iota(e)\in V$.   
  \end{itemize}
  \end{definition}

The vertex set and edge set of a graph $G$ are denoted by $V(G)$ and $\Lambda(G)$, respectively, and a similar notation is used for $\overline{E}(G), \overrightarrow{E}(G), \overline{L}(G), \overrightarrow{L}(G)$ and $S(G)$.
Since we can distinguish vertices from edges, and directed edges from the undirected ones, we assume without loss of generality that colours of vertices, of directed edges and of undirected ones are different. With a slight abuse of terminology we speak about {\em directed} and {\em undirected} edge-colours (i.e., those used on directed or on undirected edges, respectively). Note that we allow directed loops and directed normal edges to have the same colour, as well as undirected normal edges, undirected loops and semi-edges. 
Edges with the same value of the incidence function are called {\em parallel}. A graph is called {\em simple} if it has no parallel edges, no pair of opposite directed normal edges, no directed and undirected edges incident with the same pair of vertices, no loops and no semi-edges.

Loops are considered as cycles of length 1, and analogously, pairs of edges between two vertices as cycles of length 2; if they are directed then they shall use opposite directions.
A graph is {\em bipartite} if its vertex set can be partitioned into two independent (i.e., edgeless) subsets; in particular, bipartite graphs contain no loops nor semi-edges. A connected graph is a {\em tree} if it contains no loops, no semi-edges, no two normal edges incident with the same pair of vertices and no cycles (directed or undirected).   

For an undirected edge-colour $\alpha$, the $\alpha$-degree of a vertex $u$, denoted by $\mbox{deg}^{\alpha}u$, is defined as the sum of the numbers of normal edges and semi-edges of colour $\alpha$ incident with $u$ plus twice the number of loops, also of colour $\alpha$, incident with $u$. For a directed edge-colour $\alpha$, we talk about the $\alpha$-indegree and $\alpha$-outdegree of $u$, denoted by $\mbox{deg}^{\alpha}_-u$ and $\mbox{deg}^{\alpha}_+u$, respectively, which are defined in an analogous way. A graph is {\em regular} if all vertices have the same $\alpha$-degree, $\alpha$-indegree and $\alpha$-outdegree for all edge-colours $\alpha$.

When talking about a disjoint union of graphs, we assume that the graphs are vertex (and therefore also edge) disjoint. 

The following definition introduces the central notion of this paper.

\begin{definition}\label{def:graph-cover}
Let $G$ and $H$ be connected graphs coloured by the same sets of colours. A {\em covering projection} from $G$ to $H$ is a pair of colour-preserving mappings $f_V:V(G)\to V(H)$, $f_E:\Lambda(G)\to \Lambda(H)$ such that:
\begin{itemize}
  \item the preimage of an undirected normal edge of $H$ incident with vertices $u,v\in V(H)$ is a perfect matching in $G$ spanning $f^{-1}(u)\cup f^{-1}(v)$, each edge of the matching being incident with one vertex of $f^{-1}(u)$ and with one vertex of $f^{-1}(v)$;
  \item the preimage of a directed normal edge of $H$ leading from a vertex $u\in V(H)$ to a vertex $v\in V(H)$ is a perfect matching in $G$ spanning $f^{-1}(u)\cup f^{-1}(v)$, each edge of the matching being oriented from a vertex of $f^{-1}(u)$ to a vertex of $f^{-1}(v)$;

  \item the preimage of an undirected loop of $H$ incident with a vertex $u\in V(H)$ is a disjoint union of cycles in $G$ spanning $f^{-1}(u)$ --- recall that here and in the next item cycles of length 1 or 2 are also allowed;
  \item the preimage of a directed loop of $H$ incident with a vertex $u\in V(H)$ is a disjoint union of directed cycles in $G$ spanning $f^{-1}(u)$; and
  \item the preimage of a semi-edge of $H$ incident with a vertex $u\in V(H)$ is a disjoint union of semi-edges and normal edges spanning $f^{-1}(u)$ --- each vertex of $f^{-1}(u)$ being incident to exactly one semi-edge and no normal edges, or exactly one normal edge and no semi-edges, from the preimage. 
\end{itemize}
\end{definition} 

We say that $G$ {\em covers} $H$, and write $G\to H$, if there exists a covering projection from $G$ to $H$. Informally speaking, if $G$ covers $H$ via a covering projection $(f_V,f_E)$ and if
agents move along the edges of $G$ and in every moment see only the label $f_V(u)$ (or $f_E(e)$) of the vertex (edge) they are currently visiting, plus the labels of the incident edges (vertices, respectively), then the agents cannot distinguish whether they are moving through the covering graph $G$ or the target graph $H$. Mind the significant difference between undirected loops and semi-edges. The presence of an undirected loop incident with a vertex, say $u$, means that there are two ways how to move from $u$ to $u$ along this loop, while for a semi-edge, there is just one way. The same holds true for their preimages in covering projections (undirected cycles, or isolated edges). An example of a graph and a possible cover is depicted in Figure~\ref{fig:example-cover}.

\begin{figure}
\centering
\includegraphics[width=0.5\textwidth]{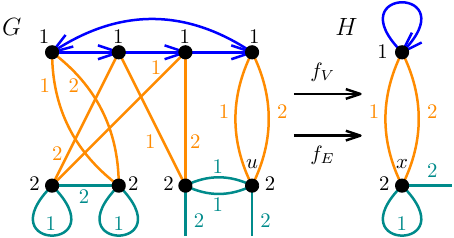}
\caption{An example of a covering projection from a graph $G$ to a graph $H$.}
\label{fig:example-cover}
\end{figure}

In~\cite{n:BFHJK21-MFCS}, a significant role of semi-edges was noted. A colour-preserving vertex-mapping $f_V: V(G)\to V(H)$ is called {\em degree-obedient} if for any edge-colour $\alpha$, any vertex $u\in V(G)$ and any vertex $x\in V(H)$, the number of edges of colour $\alpha$ that lead from $u$ to a vertex from $f_V^{-1}(x)$ in $G$ is the same as the number of edges of colour $\alpha$ leading from $f_V(u)$ to $x$ in $H$, counting those edges that may map onto each other in a covering projection.
In particular, if $x=f_V(u)$ and $H$ has $\ell$ undirected loops and $s$ semi-edges incident with $x$, and $u$ is incident with $k$ loops, $n$ normal undirected edges with both end-vertices in $f_V^{-1}(x)$ and $t$ semi-edges, then $t\le s$ and $2k+n+t=2\ell+s$; analogously for other types of edges. For an example see Figure~\ref{fig:example-cover} where for the two chosen vertices $u$ and $x$ and cyan edges we have $\ell=s=t=1$, $k=0$ and $n=2$. It is proven in~\cite{n:BFHJK21-MFCS} that every degree-obedient vertex-mapping extends to a covering projection if $H$ has no semi-edges, or when $G$ is bipartite. 

It is well known that the preimages of any two vertices in a connected graph have the same size~\cite{k:GT87,k:Reidemeister32}. For disconnected graphs, we add this requirement to the definition (as argued in~\cite{n:BFJKS24-DAM}).   

\begin{definition}\label{def:disconnected-cover}
Let $G$ and $H$ be (not necessarily connected) graphs and let  $f=(f_V,f_E) \colon G\to H$ be a pair of incidence-compatible colour-preserving mappings.
Then $f$ is a {\em  covering projection} of $G$ to $H$ if  for each component $G_i$ of $G$, the restricted mapping $f|_{G_i}:G_i \to H$ is a covering projection of $G_i$ onto some component of $H$, and for every two vertices $u,v\in V(H)$, $|f^{-1}(u)|=|f^{-1}(v)|$.
\end{definition}

Another notion we need to recall is that of the {\em degree partition} of a graph. This is a standard notion for simple undirected graphs, cf. \cite{n:CorneilG70}, and it can be naturally generalized to graphs in general. A partition of the vertex set of a graph $G$ is {\em equitable} if every two vertices of the same class of the partition
\begin{enumerate}
  \item have the same colour, and 
  \item have the same number of neighbours along edges of the same colour in every class (including its own). 
\end{enumerate}
The {\em degree partition} of a graph is then the coarsest equitable partition.  It can be found in polynomial time, and moreover, a canonical linear ordering of the classes of the degree partition comes out from the algorithm. Let $V(G)=\bigcup_{i=1}^kV_i$ be the degree partition of $G$, in the canonical ordering. The {\em degree refinement matrix} of $G$ is a $k\times k$ matrix $M_G$ whose entries are vectors indexed by edge-colours expressing that every vertex $u\in V_i$ has $M_{i,j,c}$ neighbours in $V_j$ along edges of colour $c$    (if $i=j$ and $c$ is a colour of directed edges, then  every vertex $u\in V_i$ has $M_{i,j,c}$ in-neighbours and $M_{i,j,c}$ out-neighbours in $V_i$ along edges of colour $c$). The following is proven in \cite{n:KPT97a} for graphs without semi-edges, and in~\cite{FS25} the extension to graphs with semi-edges.

\begin{proposition}\label{prop:degpart}
Let $G$ and $H$ be graphs and let $V(G)=\bigcup_{i=1}^kV_i$ and $V(H)=\bigcup_{i=1}^{\ell}W_i$ be the degree partitions of their vertex sets, in the canonical orderings. If $G$ covers $H$, then $k=\ell$, the degree refinement matrices of $G$ and $H$ are equal, and for any covering projection $f:G\to H$, $f(V_i)=W_i$ holds true for every $i=1,2,\ldots,k$. 
\end{proposition}

The classes of the degree partition are further referred to as {\em blocks}.
Once we have determined the degree partition of a graph, we re-colour the vertices so that vertices in different blocks are distinguished by vertex-colours (representing the membership to blocks), and recolour and de-orient the edges so that edges connecting vertices from different blocks are undirected and so that for any edge-colour, either all edges of this colour belong to the same block, or they are connecting vertices from the same pair of blocks. The degree partition remains unchanged after such a re-colouring.  

A {\em \bg} of $G$ is a subgraph $G'$ of $G$ whose vertex set is the union of some blocks of $G$, and such that for every edge-colour $\alpha$, $G'$ either contains all edges of colour $\alpha$ that $G$ contains, or none. A block graph $G'$ of $G$ is {\em induced} if $G'$ contains all edges of $G$ on the vertices of $V(G')$. A block graph is {\em monochromatic} if it contains edges of at most one colour. A {\em \ubg} is a block graph whose vertex set is a single block of $G$. An {\em \interbg} of $G$ is a block graph whose vertices belong to two blocks of $G$, and each of its edges is incident with vertices from both blocks (i.e., with one vertex from each block).    

Since a graph covering is defined as a local bijection, it maintains vertex degrees. In particular, vertices of degree one are mapped onto vertices of degree one and once we choose the image of such a vertex, the image of its neighbour is uniquely determined. Applied inductively, this proves the well known fact that the only connected cover of a (rooted) tree is an isomorphic copy of the tree itself~\cite{k:Reidemeister32}. (Note here, that by definition a tree is a connected graph that does not contain cycles, parallel edges, oppositely oriented directed edges, loops, and semi-edges.) As a special case, the only connected cover of a path is the path itself. These observations are the basis of the following {\em\dar} which has been introduced in~\cite{n:KPT97a} for graphs without semi-edges, and generalized to graphs with semi-edges in~\cite{n:BFJKS24-DAM}. 

\begin{definition}\label{def:reduction}{\em (\DAR)}
Let $G$ be a connected graph not isomorphic to a tree.

\begin{enumerate}
\item Determine  all vertices that belong to cycles in $G$ or that are incident with semi-edges or that lie on paths connecting the aforementioned vertices.
Determine all maximal induced subtrees pending on these vertices.

Determine the isomorphism types of these subtrees, introduce a new vertex colour for each isomorphism type, delete each subtree and colour its root by the colour corresponding to the isomorphism type of the deleted tree. In this way we obtain a graph with minimum degree at least 2 (or a single-vertex graph). 

\item If the obtained coloured graph is a path or a cycle, quit. 

\item Determine all maximal paths with at least one end-vertex of degree greater than 2 and all inner vertices being of degree exactly 2. 
Determine all colour patterns of the sequences of vertex colours, edge-colours and edge directions along such paths, and introduce a new colour for each such pattern. Replace each such path by a new edge of this colour as follows:
\begin{enumerate}
\item If both end-vertices of the path are of degree greater than 2 and the colour pattern, say $\pi$, is symmetric, the path gets replaced by an undirected edge (or loop) of colour $\pi$.

\item If both end-vertices of the path are of degree greater than 2 and the colour pattern $\pi$ is asymmetric, the path gets replaced by a directed edge (or loop) of colour $\pi$.

\item If the path ends with a semi-edge (the other end of the path must be a vertex of degree greater than 2), replace it by a semi-edge incident with its end-vertex of degree greater than 2, and colour it with colour $\pi\alpha$, where $\pi$ is the colour pattern along the path without the ending semi-edge, and $\alpha$ is the colour of the semi-edge. In this case, consider the colours corresponding to $\pi\alpha$ (on one sided open paths) and $\pi\alpha\pi^{-1}$ on symmetric paths ending with vertices of degree greater than 2 on both sides, as the same colour (this enables a path of colour pattern $\pi\alpha\pi^{-1}$ be mapped on the one sided open path in a covering projection).
\end{enumerate}
\end{enumerate}

Denote the resulting graph by $G^r$. Note that $G^r$ is a path or a cycle (if Step~2 was not performed) or has minimum degree greater than 2.
\end{definition}
 
\begin{figure}
\centering
\includegraphics[width=0.7\textwidth,page=1]{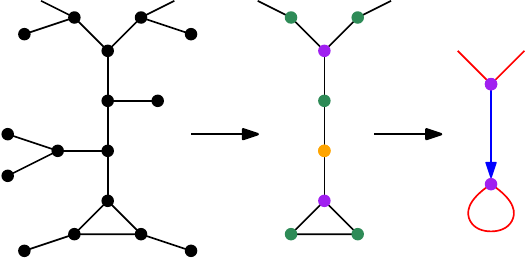}
\caption{An example of the application of the {\dar}.
}
\label{fig:reduction}
\end{figure}

See an example of such a reduction in Figure~\ref{fig:reduction}.
 
The reduced graph can be constructed in polynomial time. The usefulness of this reduction is observed in~\cite{n:KPT97a} and later on in \cite{n:BFJKS24-DAM}:

\begin{observation}\label{obs:reduction}
Given graphs $G$ and $H$, perform the {\dar} on both of them simultaneously. Then $G\to H$ if and only if $G^r\to H^r$.
\end{observation}  

Finally, for a subset $W\subseteq V(G)$, we denote by $G[W]$ the subgraph induced by $W$. If $\alpha$ is an edge-colour, then $G^{\alpha}$ denotes the spanning subgraph of $G$ containing exactly the edges of colour $\alpha$.  

\subsection{Our results}\label{subsec:ourresults}

In order to describe the results, we introduce the formal notation of certain small graphs. We denote by

\begin{itemize}
\item $F(b,c)$ the one-vertex graph with $b$ semi-edges and $c$ loops;
\item $FD(c)$ the one-vertex graph with $c$ directed loops;
\item $W(k,m,\ell,p,q)$ the two-vertex graph with $\ell$ parallel undirected edges joining its two vertices and with $k$ ($q$) semi-edges and $m$ ($p$) undirected loops incident with one (the other one, respectively) of its vertices;
\item $WD(m,\ell,m)$ the directed two-vertex graph with $m$ directed loops incident with each of its vertices, the two vertices being connected by $\ell $ directed edges in each direction;
\item $FF(c)$ the two-vertex graph connected by $c$ parallel undirected edges, with the two vertices being distinguishable to belong to different blocks;
\item $FW(b)$ the three-vertex graph with bundles of $b$ parallel edges connecting one vertex to each of the remaining two; and
\item $WW(b,c)$ the graph on four vertices obtained from a 4-cycle by replacing the edges of a perfect matching by bundles of $b$ parallel edges, and replacing the edges of the complementary matching by bundles of $c$ parallel edges,  the two independent sets of size 2  belonging to different blocks.  
\end{itemize}

Edges of all of these graphs are uncoloured (or, equivalently, monochromatic). We shall only consider $W$ graphs having $k+2m=2p+q$. See the illustration in Figure~\ref{fig:graphdef} for the graphs defined here.

\begin{figure}
\centering
\includegraphics[width=\textwidth]{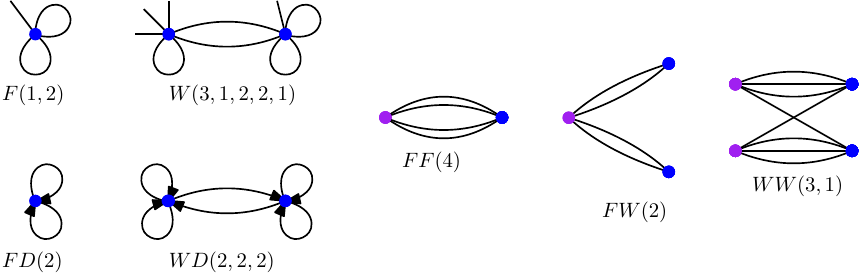}
\caption{Examples of the considered small graphs.}
\label{fig:graphdef}
\end{figure} 

\begin{definition}\label{def:key-graphs}
A regular monochromatic uniblock graph with at most two vertices is called 
\begin{itemize}
\item {\em harmless} if it is isomorphic to  $F(b,0)$, $b\le 2$, $F(1,c)$, $F(0,c)$, $FD(c)$, $W(2,0,0,0,2)$, $W(2,0,0,1,0)$, $W(0,c,0,c,0)$, $W(1,c,0,c,1)$, $W(0,0,c,0,0)$, $W(1,0,1,0,1)$, $WD(c,0,c)$, $WD(0,c,0)$, $WD(1,1,1)$ ($c$ being an arbitrary non-negative integer),
\item {\em harmful} if it is isomorphic to $F(b,c)$ such that $b\ge 2$ and $b+c\ge 3$, or to $W(k,m,\ell,p,q)$ such that $\ell\ge 1$ and $k+2m+\ell=q+2p+\ell\ge 3$, or to the disjoint union of $F(b,c)$ and $F(b',c')$ such that at least one of them is harmful, or to $WD(c,b,c)$ such that $b\ge 1, c\ge 1$ and $b+c\ge 3$. 
\end{itemize} 
A monochromatic \interbg\ is called
\begin{itemize}
\item {\em harmless} if it is isomorphic to  $FF(c)$ or $WW(0,c)$ (with $c$ being an arbitrary non-negative integer), or to $FW(0)$, $FW(1)$, or $WW(1,1)$,
\item {\em dangerous} if it is isomorphic to $FW(2)$, and 
\item {\em harmful} if it is isomorphic to $FW(c)$ for $c\ge 3$, or to $WW(b,c)$ such that $b\ge 1$, $c\ge 1$ and $b+c\ge 3$.
\end{itemize}
\end{definition}

The maximal harmless monochromatic uniblock and interblock graphs are depicted in Figure~\ref{fig:polycases}.

Note that under the assumption that each degree partition equivalence class has size at most two, every monochromatic \ubg\ as well as every monochromatic \interbg\ falls in exactly one of the above described categories. The choice of the terminology is explained by the following theorem.

\begin{theorem}\label{thm:main} 
Suppose all blocks of a connected graph $H$ have sizes at most 2. Then the following statements hold true:
\begin{enumerate}
\item If all monochromatic uniblock and interblock graphs of $H$ are harmless, then the {\sc $H$-Cover} problem is solvable in polynomial time (for arbitrary input graphs).
\item If at least one of the monochromatic uniblock or interblock graphs of $H$ is harmful, then the {\sc $H$-Cover} problem is NP-complete  even for simple input graphs. 
\item If the minimum degree of $H$ is greater than 2 and $H$ contains a dangerous monochromatic \interbg, then the {\sc $H$-Cover} problem is NP-complete even for simple input graphs.  
\end{enumerate}
\end{theorem}

Observe that Theorem~\ref{thm:main} implies that {\sc $H$-Cover} is polynomial-time solvable if and only if every monochromatic \ubg\ defines a polynomial-time solvable instance and the monochromatic \interbg s are such that either each vertex has at most one neighbour, or each vertex has degree at most two. For the \interbg s, this is also very close to saying that each monochromatic \interbg\ itself defines a polynomial-time solvable instance, but not quite. The one and only exception is the graph $FW(2)$. Indeed, {\sc $FW(2)$-Cover} is polynomial-time solvable (since it reduces to $W(0,0,1,0,0)$-Cover), but with the additional condition that all vertices have degrees greater than 2, the presence of $FW(2)$ in $H$ leads to NP-completeness of {\sc $H$-Cover} (this is shown in detail in Section~\ref{sec:NPc}).

\section{Proof of Theorem~\ref{thm:main} --- polynomial cases}\label{sec:poly}

\begin{figure}
\centering
\includegraphics[width=\textwidth]{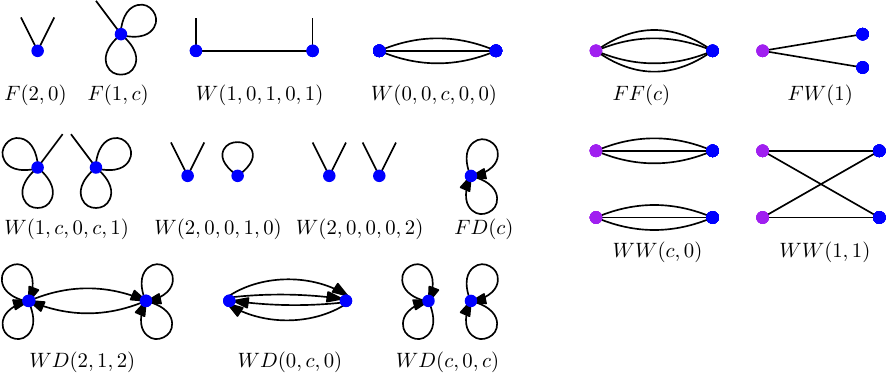}
\caption{The maximal harmless monochromatic uniblock (left) and interblock (right) graphs ($c$ is an arbitrary non-negative integer).}
\label{fig:polycases}
\end{figure} 

In this section we present the polynomial algorithm that proves Part 1 of Theorem~\ref{thm:main}. The algorithm clearly runs in polynomial time. We first describe an informal overview.

\medskip\noindent
{\bf Algorithm}
\begin{enumerate}
  \item Compute the degree partitions of $G$ (the input graph) and $H$ (the target graph). Reorder the equivalence classes $W_i$ of the degree partition of $H$ so that $W_1,\ldots,W_s$ are singletons and $W_{s+1},\ldots, W_k$ contain two vertices each, and reorder the degree partition equivalence classes $V_i$ of the input graph $G$ accordingly. Denote further, for every $i=1,\ldots,s$, by $a_i$ the vertex in $W_i$, and, for every $i=s+1,\ldots,k$, by $b_i,c_i$ the vertices of $W_i$.
  \item Check that the degree refinement matrices of $G$ and $H$ are indeed the same.
  \item Decide if the edges within $G[V_i]$ can be mapped onto the edges of $H[W_i]$ to form a covering projection, for each $i=1,2,\ldots,s$. (This step amounts to checking degrees and the numbers of semi-edges incident with the vertices, as well as checking that monochromatic  subgraphs contain perfect matchings in case of semi-edges in the target graph.)
  \item Preprocess the two-vertex equivalence classes $W_i$, $i=s+1,\ldots,k$ when $H[W_i]$ contains semi-edges (this may impose conditions on some vertices of $V_i$, whether they can map on $b_i$ or $c_i$).
  \item Using {\sc 2-Sat}, find a degree-obedient vertex mapping from $V_i$ onto $W_i$ for each $i=s+1,\ldots,k$, which fulfills the conditions observed in Step~4. (For every vertex $u\in V_i$, introduce a variable $x_u$ with the interpretation that $x_u$ is {\sf true} if $u$ is mapped onto $b_i$ and it is {\sf false} when $x_u$ is mapped onto $c_i$. The harmless block graphs are such that either all neighbours of a vertex $u$ must be mapped onto the same vertex, and thus the value of the corresponding variables are all the same (e.g., for  $WW(0,c)$), or $u$ has exactly two neighbours which should map onto different vertices (e.g., for $WW(1,1)$) meaning that the corresponding variables must get opposite values. All the situations that arise from harmless block graphs can be described by clauses of size 2.)
  \item Complete the covering projection by defining the mapping on edges in case a degree-obedient vertex mapping was found in Step 5, or conclude that $G$ does not cover $H$ otherwise. (The existence and polynomial time constructability of covering projections from degree-obedient vertex mappings for such instances have been proven in~\cite{n:BFHJK21-MFCS}.)
\end{enumerate}

We now describe the steps of the algorithm in detail. Together with the description we provide arguments for its correctness.   

\smallskip\noindent
{\em Steps 1 and 2. Checking the degree refinement matrices.} The degree partitions and the corresponding degree refinement matrices can be constructed in polynomial time. And having the same degree refinement matrices is a necessary condition for $G$ to cover $H$.

\smallskip\noindent
{\em Step 3. Checking the singleton equivalence classes.} For all $i=1,2,\ldots,s$ and for every colour $\alpha$ that appears on some of the edges of $H[W_i]$, do the following check. (Note that for each $u\in V_i$, $1\le i\le s$, the image of $u$ under a covering projection $f$ is uniquely defined, it must be $f(u)=a_i$.)

\smallskip\noindent
{\em Subcase 3A.} If $H[W_i]^{\alpha}$ has at least 2 semi-edges, it must be isomorphic to $F(2,0)$ and it is covered by $G[V_i]^{\alpha}$ if and only if $G[W_i]^{\alpha}$ is a disjoint union of cycles of even length and open paths. Output ``G does not cover H'' if this is not the case, and proceed otherwise.

\smallskip\noindent
{\em Subcase 3B.} Suppose $H[W_i]^{\alpha}$ has exactly 1 semi-edge (i.e., it is isomorphic to $F(1,c)$ for some $c$). If $G[V_i]^{\alpha}$ has a vertex with two or more semi-edges, output ``G does not cover H''. Otherwise, let $V_i^{\alpha}$ be the vertices of $V_i$ incident to a semi-edge of colour $\alpha$. Then $H[W_i]^{\alpha}$ is covered by $G[V_i]^{\alpha}$ if and only if $G[V_i\setminus V_i^{\alpha}]^{\alpha}$ has a perfect matching. This can be checked in polynomial time. Output ``G does not cover H'' if $G[V_i\setminus V_i^{\alpha}]^{\alpha}$ does not have a perfect matching, and proceed otherwise.

\smallskip\noindent
{\em Subcase 3C.} If $H[W_i]^{\alpha}$ has no semi-edges (i.e., it is isomorphic to $F(0,c)$ or to $FD(c)$ for some $c$), output ``G does not cover H'' if $G[V_i]^{\alpha}$ has at least one semi-edge, and  proceed  otherwise (in both cases $G[V_i]^{\alpha}$ 
covers $H[W_i]^{\alpha}$ by Lemma~\ref{lem:loops}).

\smallskip\noindent
{\em Step 4. Preprocessing the doublets.} 
For each $i=s+1,\ldots,k$ and for every colour $\alpha$ that appears on some of the edges of $H[W_i]$, do the following. Note that now for each $u\in V_i$, there are two possibilities for the image of $u$ under a covering projection $f$, namely $f(u)=b_i$ or $f(u)=c_i$. For each such a vertex $u$ we introduce a Boolean variable $x_u$ with the intended truth valuation $\varphi(x_u)={\sf true}$ iff $f(u)=b_i$. 

\smallskip\noindent
{\em Subcase 4A.} Suppose each vertex of  $H[W_i]^{\alpha}$ is incident with two semi-edges, i.e., this monochromatic block graph is isomorphic to $W(2,0,0,0,2)$. Then the components of $G[V_i]^{\alpha}$ are cycles and open paths. The vertices of each component map onto the same vertex from $W_i$, and the edges can map onto the edges of $F(2,0)$ if and only if their images alternate at every vertex of the component. Thus open paths are always fine, but odd cycles do not allow a covering projection onto $F(2,0)$. Hence output ``G does not cover H'' if at least one component of $G[V_i]^{\alpha}$ is an odd cycle, and proceed otherwise.

\smallskip\noindent
{\em Subcase 4B.} Suppose $H[W_i]^{\alpha}$ is isomorphic to $W(2,0,0,1,0)$, i.e., this monochromatic block graph is the disjoint union of $F(2,0)$ and $F(0,1)$. Suppose the two semi-edges are incident with $b_i$ and the loop with $c_i$ (the other case is symmetric). The components of $G[V_i]^{\alpha}$ are cycles and open paths. The vertices of a component all map to the same vertex in the intended covering projection $f$. Thus the vertices of open paths must all map onto $b_i$ and the vertices of odd cycles must map onto $c_i$. Hence we set $\varphi(x_u)={\sf true}$ if $u$ belongs to an open path in $G[V_i]^{\alpha}$, and we set $\varphi(x_u)={\sf false}$ if $u$ belongs to an odd cycle in $G[V_i]^{\alpha}$.

\smallskip\noindent
{\em Subcase 4C.} Suppose the vertices of $H[W_i]$ are incident with one semi-edge each, and this monochromatic block graph is disconnected, i.e., it is isomorphic to $W(1,c,0,c,1)$ for some $c\ge 0$. 
If $G[V_i]^{\alpha}$ has a vertex with two or more semi-edges, output ``G does not cover H''. Otherwise, let $V_i^{\alpha}$ be the vertices of $V_i$ incident to a semi-edge of colour $\alpha$. Then $H[W_i]^{\alpha}$ can only be covered by $G[V_i]^{\alpha}$  if $G[V_i\setminus V_i^{\alpha}]^{\alpha}$ has a perfect matching. This can be checked in polynomial time. Output ``G does not cover H'' if $G[V_i\setminus V_i^{\alpha}]^{\alpha}$ does not have a perfect matching, and proceed otherwise.

\smallskip\noindent
{\em Subcase 4D.} Suppose the vertices of $H[W_i]$ are incident with one semi-edge each, and this monochromatic block graph is connected, i.e., it is isomorphic to $W(1,0,1,0,1)$. If $G[V_i]^{\alpha}$ has a vertex with two or more semi-edges, output ``G does not cover H'', and proceed otherwise.

\smallskip\noindent
{\em Step 5. Setting up the 2-SAT formula.} Since the variables $x_u$ are created for vertices $u$ that must be mapped onto vertices of the doublet equivalence classes of $H$, only the monochromatic block graphs of the doublets, and interblock graphs that are induced by at least one doublet block, will come into play. First we consider the monochromatic block graphs. For every $i=s+1,\ldots,k$ and every colour $\alpha$ that is used on some edges of $H[W_i]$, do the following.

\smallskip\noindent
{\em Subcase 5A.} Suppose the monochromatic block graph $H[W_i]^{\alpha}$ is disconnected, i.e., it is isomorphic to $W(2,0,0,0,2)$, $W(2,0,0,1,0)$, $W(1,c,0,c,1)$, $W(0,c,0,c,0)$, or $WD(c,0,c)$ for some $c\ge 0$. Then any two adjacent vertices, say $u$ and $v$, of $G[V_i]^{\alpha}$ must be mapped onto the same vertex of $W_i$ by the intended covering projection $f$, and thus we introduce a 2-clause $$(\varphi(x_u)\Leftrightarrow\varphi(x_v)).$$

\smallskip\noindent
{\em Subcase 5B.} Suppose the monochromatic block graph $H[W_i]^{\alpha}$ is connected and bipartite, i.e., it is isomorphic to $W(0,0,c,0,0)$ or $WD(0,c,0)$ for some $c\ge 1$. Then any two adjacent vertices, say $u$ and $v$, of $G[V_i]^{\alpha}$ must be mapped onto different vertices of $W_i$ by the intended covering projection $f$, and thus we introduce a 2-clause $$(\varphi(x_u)\Leftrightarrow\neg\varphi(x_v)).$$

\smallskip\noindent
{\em Subcase 5C.} Suppose the monochromatic block graph $H[W_i]^{\alpha}$ is isomorphic to $W(1,0,1,0,1)$. We have already ruled out the case that $G[V_i]^{\alpha}$ contains a vertex incident with two semi-edges, and thus every vertex of $V_i$ has either one or two neighbours in $G[V_i]^{\alpha}$. If $u\in V_i$ has exactly one neighbour, say $v$, in $G[V_i]^{\alpha}$, it is incident to a semi-edge which must map onto a semi-edge of $W(1,0,1,0,1)\simeq H[W_i]^{\alpha}$, and the edge $uv$ must map onto the ordinary edge of $W(1,0,1,0,1)\simeq H[W_i]^{\alpha}$. Thus $u$ and $v$ are mapped onto different vertices of $W_i$, and this justifies introducing a 2-clause $$(\varphi(x_u)\Leftrightarrow\neg\varphi(x_v)).$$ If $u$ is adjacent to two vertices, say $v$ and $w$, in $G[V_i]^{\alpha}$, then one of the edges $uv,uw$ must map onto the ordinary edge of $W(1,0,1,0,1)\simeq H[W_i]^{\alpha}$, while the other one must map onto a semi-edge. Hence one of the vertices $v,w$ is mapped onto the same vertex as $u$, and the other one on the other one.  This justifies introducing a 2-clause $$(\varphi(x_w)\Leftrightarrow\neg\varphi(x_v)).$$

\smallskip\noindent
{\em Subcase 5D.} Suppose $\alpha$ is a directed colour and the monochromatic block graph $H[W_i]^{\alpha}$ is isomorphic to $WD(1,1,1)$. Then every vertex $u$ of $V_i$ has two outgoing neighbours, say $v,w$, and two incoming neighbours, say $y,z$, in $G[V_i]^{\alpha}$. One of the outgoing edges $uv,uw$ is mapped onto a loop and the other one onto an ordinary edge in $H[V_i]^{\alpha}\simeq WD(1,1,1)$, and hence one of the vertices $v,w$ is mapped onto the same vertex as $u$, and the other one on the other one. Similarly for the incoming neighbours. This justifies introducing the 2-clauses
$$(\varphi(x_w)\Leftrightarrow\neg\varphi(x_v)) \ \ \textrm{and} \ \ (\varphi(x_z)\Leftrightarrow\neg\varphi(x_y)).$$

For the interblock graphs, we process all pairs $i\neq j$, $1\le i\le k, s+1\le j\le k$, and all colours $\alpha$ that are used on some edges of the interblock graph $H[W_i\cup W_j]^{\alpha}$:

\smallskip\noindent
{\em Subcase 5E.} Suppose $i\le s$, i.e., $W_i$ is a singleton and $H[W_i\cup W_j]^{\alpha}$ is isomorphic to $FW(1)$. Every vertex $u\in V_i$ has two neighbours, say $v$ and $w$, in $G[V_i\cup V_j]^{\alpha}$, and these neighbours must be mapped onto different vertices of $H[W_j]$.  This justifies introducing a 2-clause $$(\varphi(x_w)\Leftrightarrow\neg\varphi(x_v)).$$

\smallskip\noindent
{\em Subcase 5F.} Suppose $i> s$, i.e., $W_i$ is a doublet, and suppose that $H[W_i\cup W_j]^{\alpha}$ is connected, i.e., it is isomorphic to $WW(1,1)$. Every vertex $u\in V_i$ has two neighbours, say $v$ and $w$, in $G[V_i\cup V_j]^{\alpha}$, and these neighbours must be mapped onto different vertices of $H[W_j]$. But also every vertex $v\in V_j$ has two neighbours, say $u$ and $y$, in $V_i$, and these must be mapped onto different vertices of $W_i$. This justifies introducing two 2-clauses
$$(\varphi(x_w)\Leftrightarrow\neg\varphi(x_v)) \ \ \textrm{and} \ \
(\varphi(x_u)\Leftrightarrow\neg\varphi(x_y)).$$

{\em Subcase 5G.} Suppose $i> s$, i.e., $W_i$ is a doublet, and suppose that $H[W_i\cup W_j]^{\alpha}$ is disconnected, i.e., it is isomorphic to $WW(c,0)$ for some $c>0$. Then the image of a vertex $u\in V_i$ is uniquely determined by the mapping of any of its neighbours in $G[V_i\cup V_j]^{\alpha}$. For any two adjacent vertices $u\in V_i, v\in V_j$, we introduce the 2-clause 
$$(\varphi(x_u)\Leftrightarrow\varphi(x_v))$$ if 
$H[W_i\cup W_j]^{\alpha}$ is such that $b_i$ is adjacent to $b_j$ (and $c_i$ to $c_j$), while we introduce the 2-clause
$$(\varphi(x_u)\Leftrightarrow\neg\varphi(x_v))$$
otherwise (i.e., when $b_i$ is adjacent to $c_j$ and $c_i$ to $b_j$).

\smallskip
{\em Step 6. Creating the covering projection.} We solve the 2-SAT problem of satisfiability of the formula whose construction we described in Step~3, which can be done in polynomial time. We output ``G does not cover H'' if this formula is not satisfiable, and we output ``G does cover H'' otherwise. If we want to construct a covering projection in the latter case, we take a satisfying truth valuation $\varphi$ and define the vertex mapping by setting

$$f(u)=\left\lbrace 
\begin{array}{ll}
a_i & \quad\mbox{if $u\in V_i$ and $i\le s$}\\
b_i & \quad\mbox{if $u\in V_i$ and $i\ge s+1$ and $\varphi(x_u)={\sf true}$}\\
c_i & \quad\mbox{if $u\in V_i$ and $i\ge s+1$ and $\varphi(x_u)={\sf false}$}.
\end{array} \right.$$

It is easy to see that this is a degree-obedient vertex mapping from $G$ onto $H$. It remains to define the edge mapping. This is forced in the case of edges $uv\in E(G)$ such that there is only one edge connecting $f(u)$ and $f(v)$ in $H$, of the same colour, say $\alpha$, as the colour of $uv$ in $G$. If there are multiple edges connecting $f(u)$ and $f(v)$ in $H^{\alpha}$, then $G[f^{-1}(f(u))\cup f^{-1}(f(v))]^{\alpha}$ is a regular graph and its edges can be distributed so that they cover the multiple edge connecting $f(u)$ and $f(v)$. This follows from the following lemmas, which are just reformulations of the famous K\"{o}nig-Hall and Petersen theorems on factors of graphs. Lemma~\ref{lem:loops} applies to the case of $f(u)=f(v)$, while Lemma~\ref{lem:bars} applies to the case when $f(u)\neq f(v)$. (A directed graph is called $k$-in-regular if the in-degree of every vertex is $k$, it is called $k$-out-regular if the out-degree of every vertex is $k$, and it is called $k$-in-$k$-out-regular if it is $k$-in-regular and $k$-out-regular.)

\begin{lemma}\label{lem:bars}
Every $k$-regular undirected bipartite graph covers $FF(k)$, and a covering projection can be found in polynomial time.
\end{lemma}

\begin{lemma}\label{lem:loops}
Every $2k$-regular undirected graph with no semi-edges covers $F(k)$. Every $k$-in-$k$-out-regular directed graph covers $FD(k)$. Respective covering projections can be found in polynomial time.
\end{lemma}

As optimizing the running time of the algorithm is not our goal, we omit an explicit running time analysis. We just mention that all steps can be performed in polynomial time.

\section{Proof of Theorem~\ref{thm:main} --- NP-hard cases}\label{sec:NPc}

\subsection{Monochromatic uniblock graphs}

\begin{proposition}\label{prop:NP-block}
The following problems are NP-complete even for simple input graphs
\begin{itemize}
\item {\sc $F(b,c)$-Cover} for every $b\ge 2$ and $c\ge 0$ such that  $b+c\ge 3$,
\item {\sc $W(k,m,\ell,p,q)$-Cover} for $k,m,\ell,p,q\ge 0$ such that $k+2m=2p+q$, and either $\ell\ge 1$ and $k+2m+\ell\ge 3$, or $\ell=0$, $k+m\ge 3$ and $\min\{k,q\}\ge 2$, and
\item {\sc $WD(b,c,b)$-Cover} for every $b,c\ge 1$ such that $b+c\ge 3$.
\end{itemize}
\end{proposition}

\begin{proof}
The NP-completeness of {\sc $F(b,c)$-Cover} for $b\ge 2$ and $b+c\ge 3$ for simple input graphs is proven in \cite{n:BFHJK21-MFCS}, and so is the NP-completeness of {\sc $W(k,m,\ell,p,q)$-Cover} for $\ell \ge 1$ and $k+2m+\ell\ge 3$, even for simple bipartite input graphs. 
The NP-completeness of {\sc $W(k,m,\ell,p,q)$-Cover} for $\ell=0$, that is for a disconnected target graph, follows from the results of \cite{n:BFJKS24-DAM} which analyze in detail the concept and results on covering disconnected graphs. 
Note that for $\ell=0$, $W(k,m,0,p,q)$ is the disjoint union of $F(k,m)$ and $F(q,p)$, and {\sc $W(k,m,0,p,q)$-Cover} is NP-complete if and only if at least one of the problems  {\sc $F(k,m)$-Cover} and {\sc $F(q,p)$-Cover} is NP-complete.

\begin{figure}
\centering
\includegraphics[width=\textwidth]{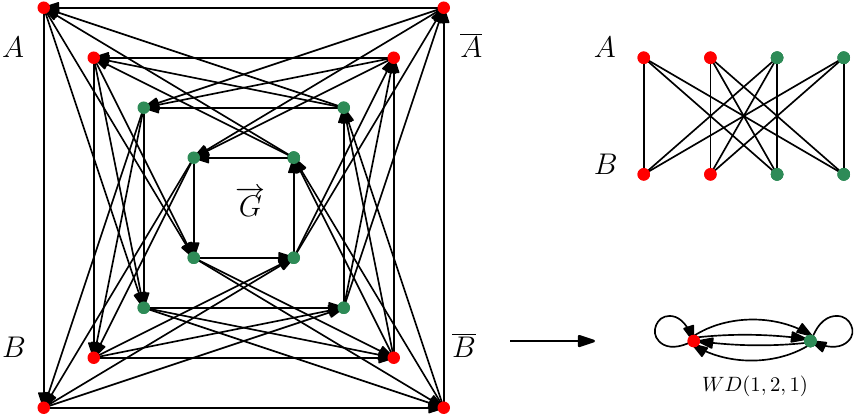}
\caption{Illustration to covering of $WD(b,c,b)$.}
\label{fig:WD-NPc}
\end{figure}

For directed graphs, the NP-completeness of {\sc $WD(b,c,b)$-Cover} for $b,c\ge 1$, $b+c\ge 3$ has already been proven in \cite{n:KPT97a}, but not for simple input graphs. However, we can argue as follows. Take a simple $(b+c)$-regular graph $G$ as input to the problem {\sc $(b,c)$-Colouring} which asks for a 2-colouring of the vertices of $G$ such that every vertex has $b$ neighbours of its own colour and $c$ neighbours of the other one. This problem is NP-complete even for bipartite graphs, as shown in \cite{n:BFHJK21-MFCS}, so assume that $V(G)=A\cup B$ are such that both $A$ and $B$ are independent sets. For every vertex $u\in V(G)$, introduce a new extra vertex $\overline{u}$, and set $\overline{A}=\{\overline{u}:u\in A\}$ and $\overline{B}=\{\overline{u}:u\in B\}$. Define a simple directed graph $\overrightarrow{G}$ by setting
$$V(\overrightarrow{G})=A\cup B\cup \overline{A}\cup \overline{B}$$ and
$$E(\overrightarrow{G})=\{uv,v\overline{u},\overline{u}\overline{v},\overline{v}u:uv\in E(G),u\in A, v\in B\}.$$
Then  $\overrightarrow{G}$ covers  $WD(b,c,b)$ if and only if $G$ allows a $(b,c)$-colouring. 

Suppose first that $\overrightarrow{G}$ covers  $WD(b,c,b)$ and let $f:\overrightarrow{G}\to WD(b,c,b)$ be a covering projection. Denote the vertices of $WD(b,c,b)$ by $r$ and $g$ and colour a vertex $x$ of $G$ {\sf green} if $f(x)=g$, and colour it {\sf red} if $f(x)=r$. Consider a vertex $u\in A$ coloured {\sf red}. Since $b$ of the $b+c$ directed edges that start in $u$ are mapped to the loops incident with $r$, $u$ has $b$ {\sf red} neighbours in $B$. The remaining $c$ outgoing edges are mapped onto the bar edges connecting $r$ and $g$, and hence $u$ has $c$ {\sf green} neighbours in $B$. Similarly for a {\sf red} vertex $v\in B$, we consider the $b+c$ incoming edges, $b$ of them being mapped onto the $b$ loops incident with $r$ and $c$ of them being mapped onto the bars connecting $r$ and $g$. Thus $v$ has $b$ {\sf red} and $c$ {\sf green} neighbours in $A$. An analogous argument applies to {\sf green} vertices, and we conclude that the the {\sf red-green} colouring of $A\cup B$ is a valid $(b,c)$-colouring of $G$.

On the other hand, suppose that $\varphi:A\cup B\to \{{\sf red, green}\}$ is a $(b,c)$-colouring of $G$. Define 
$$f(u)=f(\overline{u})=
\left\lbrace
\begin{array}{ll}
r & \mbox{ if }\varphi(u)={\sf red}\\
g  & \mbox{ if }\varphi(u)={\sf green}.
\end{array} 
\right.$$
We see that $f$ is a degree-obedient vertex mapping from $\overrightarrow{G}$ onto $WD(b,c,b)$, and since $\overrightarrow{G}$ is bipartite, the existence of a covering projection follows by arguments similar to those in the proof of Proposition~8 in \cite{n:BFHJK21-MFCS}. Cf. the illustrative example in Figure~\ref{fig:WD-NPc} for the reduction.
\end{proof}

\subsection{Monochromatic interblock graphs}

The NP-hardness for the harmful interblock graphs is claimed in the following proposition.

\begin{proposition}\label{prop:NP-interblock-easy}
The following problems are NP-complete even for simple input graphs
\begin{enumerate}
\item {\sc $FW(c)$-Cover} for every $c\ge 3$, and 
\item {\sc $WW(b,c)$-Cover} for every $b,c\ge 1$ such that  $b+c\ge 3$.
\end{enumerate}
\end{proposition}

\begin{proof}
\textit{1.} The NP-hardness of {\sc $FW(c)$-Cover} is proven in \cite{n:KTT16}, but only for multigraphs on input. To prove the stronger result, we reduce from {\sc $c$-in-$2c$-Satisfiability} of all-positive formulas with each variable occurring in exactly $c$ clauses. This problem is NP-complete for every fixed $c\ge 3$~\cite{kratochvil2003complexity}. Given a formula $\Phi$ with the set $C$ of clauses over a set $X$ of variables such that every clause contains exactly $2c$ variables, all positive, and every variable occurs in exactly $c$ clauses, the question is if the variables can be assigned values {\sf true} and {\sf false} so that every clause contains exactly $c$ variables evaluated to {\sf true}. We construct a graph $G_{\Phi}$ such that $G_{\Phi}$ covers $FW(c)$ if and only if $\Phi$ is $c$-in-$2c$-satisfiable. The graph contains one vertex $z(s)$ for every clause $s\in C$, and this vertex is adjacent to a connecting vertex $w(x,s)$ for every variable $x\in s$. Every variable $x\in X$ is represented by a variable gadget which contains $c$ connecting vertices $w(x,s), x\in s\in C$ and ensures that all of these map to the same vertex of $FW(c)$. More formally, we set
$$
\begin{array}{ll}
V(G_{\Phi})= & \{z(s):s\in C\}\cup\{w(x,s):x\in s\in C\}\\ 
\ & \cup \{u_i^{x,s}:x\in s\in C, i=1,2,\ldots,c-1\}\\
\ & \cup\{v_i^{x,j}:x\in X,i=1,2,\ldots,c-1,j=1,2,\ldots,2c-1\}, \\
E(G_{\Phi})= & \{z(s)w(x,s):x\in s\in C\}\\
\ & \cup \{w(x,s)u_i^{x,s}:x\in s\in C,i=1,2,\ldots,c-1\}\\
\ & \cup\{u_i^{x,s}v_i^{x,j}: x\in s\in C,i=1,2,\ldots,c-1,j=1,2,\ldots,2c-1 \}.
\end{array}
$$  

Let the two vertices of degree $c$ in $FW(c)$ be $r$ and $g$ again, and let the vertex of degree $2c$ be $p$. Suppose $f:G_{\Phi}\to FW(c)$ is a covering projection. The graph $G_{\Phi}$ contains $|C|+|X|\cdot c\cdot (c-1)$ vertices of degree $2c$, namely the vertices $z(s)$ and $u_i^{x,s}$. All of these must be mapped to the vertex $p$. Consider two clauses $s,t\in C$ that both contain a variable $x$. The vertices $u_1^{x,s}$ and $u_1^{x,t}$ have $2c-1$ common neighbours $v_1^{x,j}, j=1,2,\ldots,2c-1$, and since both of them must have $c$ neighbours mapped onto $r$ and $c$ neighbours mapped onto $g$, we see that their private neighbours $w(x,s)$ and $w(x,t)$ are mapped onto the same vertex. Hence a truth valuation $\varphi:X\to \{{\sf true,false}\}$ set by
$$
\varphi(x)=\left\lbrace
\begin{array}{ll}
{\sf true} & \mbox{ if } f(w(x,s))=r \mbox{ for every $s$ such that $x\in s$}, \\
{\sf false} & \mbox{ if } f(w(x,s))=g \mbox{ for every $s$ such that $x\in s$}
\end{array}
\right.
$$   
is correctly defined. Since every vertex $z(s)$ must see $c$ neighbours mapped onto $r$ and $c$ neighbours mapped onto $g$, every clause $s$ gets exactly $c$ variables evaluated to {\sf true}.

Let, on the other hand, $\varphi:X\to \{\sf true, false\}$ be a truth valuation of $\Phi$ such that every clause contains $c$ variables evaluated to {\sf true} and $c$ variables evaluated to {\sf false}. Define a vertex mapping $f:V(G_{\Phi}\to\{p,r,g\}$ by setting
$$
\begin{array}{ll}
f(z(s)) = & f(u_i^{x,s})=p \textrm{ for all } s\in C, x\in s, i=1,2,\ldots,c-1,\\ 
f(v_i^{x,1}) = & \ldots=f(v_i^{x,c})=r \mbox{ and } f(v_i^{x,c+1})=\ldots=f(v_i^{x,2c-1})=f(w(x,s))=g\\
\ & \textrm{ for all } i=1,2,\ldots,c-1, x\in X, x\in s\in C \textrm{ such that } \varphi(x)= \textrm{{\sf true}}, \\
f(v_i^{x,1}) = & \ldots=f(v_i^{x,c})=g \textrm{ and } f(v_i^{x,c+1})=\ldots=f(v_i^{x,2c-1})=f(w(x,s))=r\\
\ & \textrm{ for all } i=1,2,\ldots,c-1, x\in X, x\in s\in C \textrm{ such that } \varphi(x)= \textrm{{\sf true}}.
\end{array}
$$ 

\medskip\noindent
Every vertex of degree $2c$ is mapped onto $p$ (these are the $z(s)$ and $u_{i}^{x,s}$ vertices) and all other vertices are of degree $c$ and are mapped onto $r$ or $g$. Every vertex of degree $c$ has all neighbours mapped onto $p$. And every vertex mapped onto $p$ has exactly $c$ neighbours mapped onto $r$ and $c$ neighbours mapped onto $g$ (the $u_i^{x,s}$ vertices by definition of $f$, the $z(s)$ vertices because $\varphi$ is $c$-in-$2c$-satisfying). Thus $f$ is a degree-obedient vertex mapping, and since $G_{\Phi}$ is bipartite, it extends to a covering projection from $G_{\Phi}$ onto $FW(c)$.

\medskip\noindent
2. For the second case, note that a bipartite graph allows a $(b,c)$-colouring if and only it covers $WW(b,c)$, and thus the claim follows from Theorem~18 of \cite{n:BFHJK21-MFCS}. 
\end{proof}

Since they cover all harmful monochromatic block graphs, Propositions~\ref{prop:NP-block} and~\ref{prop:NP-interblock-easy} form the base cases of Part 2 of Theorem~\ref{thm:main}. 

\subsection{The dangerous interblock graph}

Part 3 of Theorem~\ref{thm:main} is, however, more complicated and requires a detailed analysis of the neighbourhood of a dangerous monochromatic interblock graph. This is the only case when a monochromatic block graph induces a polynomial-time solvable cover problem (indeed, {\sc $FW(2)$-Cover} is reducible to {\sc $F(0,2)$-Cover}, which is polynomial) but just the fact that all vertices have degrees greater than 2 makes the problem NP-complete.

\begin{proposition}\label{prop:NP-interblock-uneasy}
Let $H$ be a graph such that each block has at most 2 vertices, the minimum degree of $H$ is greater than 2, $H$ contains a monochromatic block graph isomorphic to $FW(2)$, but does not contain any harmful monochromatic uniblock or interblock graph.
Then $H$ contains a block graph which is reducible to one of the graphs from Figure~\ref{fig:FW(2)-caseanalysis}.
\end{proposition}

\begin{figure}
\centering
\includegraphics[width=\textwidth]{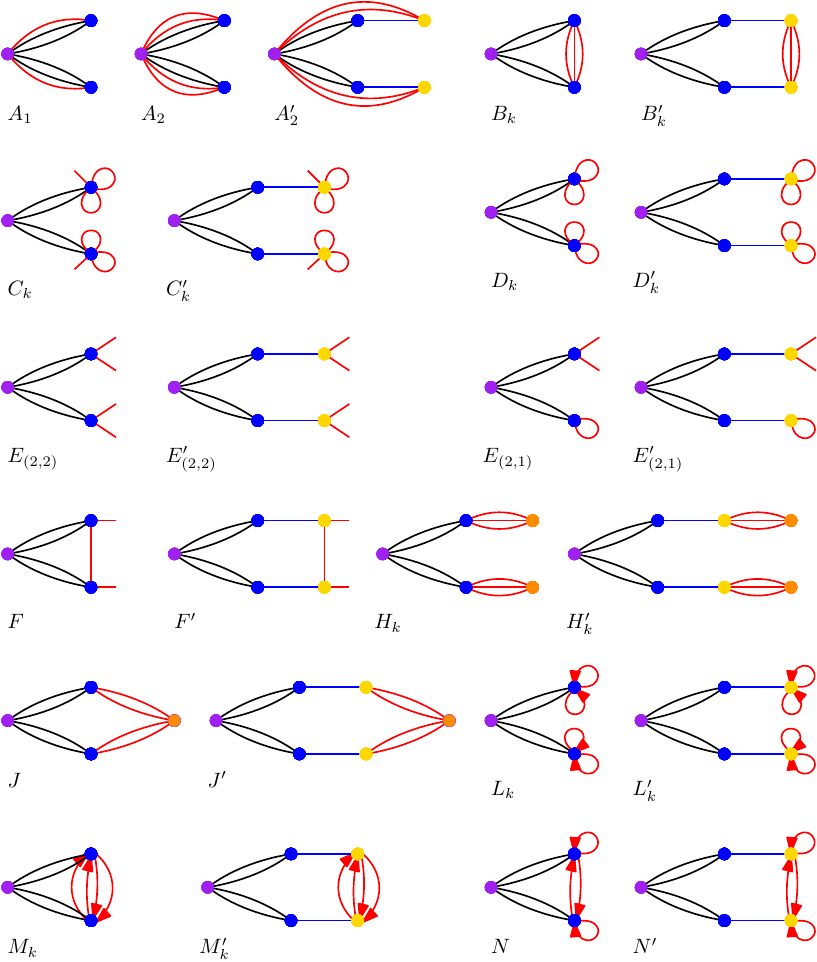}
\caption{Block graphs forced by $FW(2)$.}
\label{fig:FW(2)-caseanalysis}
\end{figure}

\begin{proof}
Let $W_1,W_2$ be the blocks such that $|W_1|=1,|W_2|=2$ and $H[W_1\cup W_2]^{\alpha}\simeq FW(2)$ for some colour $\alpha$ (in the figures, the vertex of $W_1$ is purple, the vertices of $W_2$ are blue and the colour $\alpha$ is black). Since the vertices of $W_2$ have degree at least 3 in $H$, there is another edge-colour, say $\beta$, whose edges are incident with vertices of $W_2$. We distinguish three possibilities:

\smallskip\noindent
{\em Case 1. The edges of colour $\beta$ connect the vertices of $W_2$ to the vertex of $W_1$.} Then $H[W_1\cup W_2]^{\beta}$ is isomorphic either to $FW(1)$ or to $FW(2)$  (if there were three or more parallel edges of colour $\beta$ between a vertex of $W_2$ and the vertex of $W_1$, $H[W_1\cup W_2]^{\beta}$ would be a harmful monochromatic \interbg ), and thus $H$ contains a block graph isomorphic to $A_1$ or to $A_2$.

\smallskip\noindent
{\em Case 2. The edges of colour $\beta$ connect the vertices of $W_2$ to themselves.} In this case $H[W_1\cup W_2]^{\beta}$ can be any nonempty harmless graph on the doublet block $W_2$, and thus $H$ contains a block graph isomorphic to one of the graphs $B_k, k\ge 1$, $C_k, k\ge 0$, $D_k, k\ge 1$, $E_{(2,2)}$, $E_{(2,1)}$, $F$, $L_k, k\ge 1$, $M_k, k\ge 1$, or $N$. (For the graphs parameterized by $k$, $k$ expresses the multiplicity of undirected ordinary edges (in $B_k$, and similarly in $B'_k, H_k, H'_k$), or the multiplicity of loops (in $C_k$ and $D_k$, and similarly in    $C'_k$ and $D'_k$), or the multiplicity of directed loops (in $L_k$ and $L'_k$), or the multiplicity of directed ordinary edges (in $M_k$ and $M'_k$)).
 
\smallskip\noindent
{\em Case 3. The edges of colour $\beta$ connect the vertices of $W_2$ to vertices of a new block $W_3$.} 

{\em Subcase 3A: $|W_3|=1$. } In this case, $H[W_2\cup W_3]^{\beta}\simeq FW(k)$ for $k=1$ or $k=2$. If $k=2$, $H[W_1\cup W_2\cup W_3]^{\alpha,\beta}\simeq J$. If $k=1$, the edges of colour $\beta$ form a path of length 2 connecting the vertices of $W_2$, with the vertex of $W_3$ being of degree 2, and hence $H[W_1\cup W_2\cup W_3]^{\alpha,\beta}$ is reducible to $B_1$.

{\em Subcase 3B: $|W_3|=2$. } Since the edges of colour $\beta$ do not induce a harmful \interbg, $H[W_2\cup W_3]^{\beta}$ is isomorphic either to $WW(k,0)$ or to $WW(1,1)$. If it is $WW(k,0)$ for $k\ge 3$, $H[W_1\cup W_2\cup W_3]^{\alpha,\beta}$ is isomorphic to $H_k$, while for $k=2$, $H[W_1\cup W_2\cup W_3]^{\alpha,\beta}$ is reducible to $D_1$. If $H[W_2\cup W_3]^{\beta}\simeq WW(1,1)$, $H[W_1\cup W_2\cup W_3]^{\alpha,\beta}$ is reducible to $B_2$.
 
 \smallskip\noindent
 {\em Case 4.}
The last subsubcase of 3B, when $H[W_2\cup W_3]^{\beta}\simeq WW(1,0)$, is responsible for all the primed graphs of Figure~\ref{fig:FW(2)-caseanalysis}. Vertices of $W_3$ are of degree greater than 1, and thus there must be another colour  such that edges of this colour are incident with vertices of $W_3$. We would repeat the case analysis for this edge-colour, and repeat and repeat, if necessary. This leads to the following definition. Let $t$ be the largest integer such that there exist doublet blocks $W_3,W_4,\ldots,W_t$ and colours $\beta_3,\ldots,\beta_t$ such that $H[W_{i-1}\cup W_i]^{\beta_i}\simeq WW(1,0)$ for all $i=3,4,\ldots,t$. Setting $\beta_3=\beta$, we see that $t$ is well defined. Vertices of $W_t$ have degree 1 in $H[W_{t-1}\cup W_t]^{\beta_t}$ and thus they must be incident with edges of another colour, say~$\gamma$. 

{\em Subcase 4A. The edges of colour $\gamma$ lead to a new extra block $W_{t+1}$.} If this block is a singleton, $H[W_t\cup W_{t+1}]^{\gamma}$ is isomorphic either to $FW(1)$ or $FW(2)$. In the former case, $H[\bigcup_{i=2}^{t+1}W_i]^{\beta_3,\ldots,\beta_t,\gamma}$ is a symmetric path with all inner vertices of degree 2, and thus $H[\bigcup_{i=1}^{t+1}W_i]^{\alpha,\beta_3,\ldots,\beta_t,\gamma}$ is reducible to  $B_1$ (cf. Figure~\ref{fig:FW(2)-caseanalysisIII} top left). Analogously,  $H[\bigcup_{i=1}^{t+1}W_i]^{\alpha,\beta_3,\ldots,\beta_t,\gamma}$ is reducible to  $J'$ in case of $H[W_t\cup W_{t+1}]^{\gamma}$ being isomorphic to $FW(2)$.

If $W_{t+1}$ is a doublet, $H[W_t\cup W_{t+1}]^{\gamma}$ is isomorphic to $WW(1,1)$ or to $WW(k,0)$ for some $k>1$ (it cannot be isomorphic to $WW(1,0)$ since then $t$ would not be maximal as defined). In the former case, $H[\bigcup_{i=1}^{t+1}W_i]^{\alpha,\beta_3,\ldots,\beta_t,\gamma}$ is reducible to $B'_2$ (illustrated in the second row on the right of Figure~\ref{fig:FW(2)-caseanalysisIII}). If  $H[W_t\cup W_{t+1}]^{\gamma}$ is isomorphic to $WW(2,0)$, $H[\bigcup_{i=1}^{t+1}W_i]^{\alpha,\beta_3,\ldots,\beta_t,\gamma}$ is reducible to $D'_1$ (illustrated in Figure~\ref{fig:FW(2)-caseanalysisIII} top right).  If  $H[W_t\cup W_{t+1}]^{\gamma}\simeq WW(k,0)$ for $k>2$, $H[\bigcup_{i=1}^{t+1}W_i]^{\alpha,\beta_3,\ldots,\beta_t,\gamma}$ is reducible to $H'_k$.

{\em Subcase 4B. The edges of colour $\gamma$ connect vertices inside $W_{t}$.} In this case $H[W_t]^{\gamma}$ is a harmless monochromatic \ubg\ and $H[\bigcup_{i=1}^{t}W_i]^{\alpha,\beta_3,\ldots,\beta_t,\gamma}$ is reducible to one of the graphs $B'_k, k\ge 1$, $C'_k, k\ge 0$, $D'_k, k\ge 1$, $E'_{(2,2)}$, $E'_{(2,1)}$, $F'$, $L'_k, k\ge 1$, $M'_k,k\ge 1$, or $N'$.

{\em Subcase 4C. The edges of colour $\gamma$ lead to the vertex of $W_1$.} If $H[W_1\cup W_t]^{\alpha}$ is isomorphic to $FW(1)$, $H[\bigcup_{i=1}^{t}W_i]^{\alpha,\beta_3,\ldots,\beta_t,\gamma}$ is reducible to $A_1$ (illustrated in the second row on the left of Figure~\ref{fig:FW(2)-caseanalysisIII}). Analogously, if  $H[W_1\cup W_t]^{\alpha}$ is isomorphic to $FW(2)$, $H[\bigcup_{i=1}^{t}W_i]^{\alpha,\beta_3,\ldots,\beta_t,\gamma}$ is reducible to $A'_2$.

{\em Subcase 4D. The edges of colour $\gamma$ lead to the vertices of a block $W_j$ for some $j, 2\le j\le t-1$.} Then $H[W_j\cup W_t]^{\gamma}$ is isomorphic to $WW(1,1)$ or to $WW(k,0)$ for some $k\ge 1$.

If it is isomorphic to $WW(1,0)$, there are two possibilities. Either each edge of colour $\gamma$ connects vertices within the same path, or across the paths. The block graph $H[\bigcup_{i=1}^{t}W_i]^{\alpha,\beta_3,\ldots,\beta_t,\gamma}$ is reducible to $L_1$ or $L'_1$ in the former case, and to $M_1$ or $M'_1$ in the latter one. (It is $L_1$ or $M_1$ if $j=2$ and $L'_1$ or $M'_1$ if $j>2$.) These cases are illustrated in the third row of Figure~\ref{fig:FW(2)-caseanalysisIII}.

If $H[W_j\cup W_t]^{\gamma}\simeq WW(1,1)$, then the block graph  $H[\bigcup_{i=1}^{j}W_i\cup W_t]^{\alpha,\beta_3,\ldots,\beta_j,\gamma}$ is reducible to $B'_2$ (or $B_2$, when $j=2$) as illustrated in the bottom row of Figure~\ref{fig:FW(2)-caseanalysisIII}.\\
When $H[W_j\cup W_t]^{\gamma}\simeq WW(k,0)$ for $k\ge 3$, then the block graph  $H[\bigcup_{i=1}^{j}W_i\cup W_t]^{\alpha,\beta_3,\ldots,\beta_j,\gamma}$ is reducible to $H'_k$ (or $H_k$, if $j=2$). When $H[W_j\cup W_t]^{\gamma}\simeq WW(2,0)$, the block graph  $H[\bigcup_{i=1}^{j}W_i\cup W_t]^{\alpha,\beta_3,\ldots,\beta_j,\gamma}$ is reducible to $D'_1$ (or $D_1$, if $j=2$).
\end{proof}

\begin{figure}
\centering
\includegraphics[width=\textwidth]{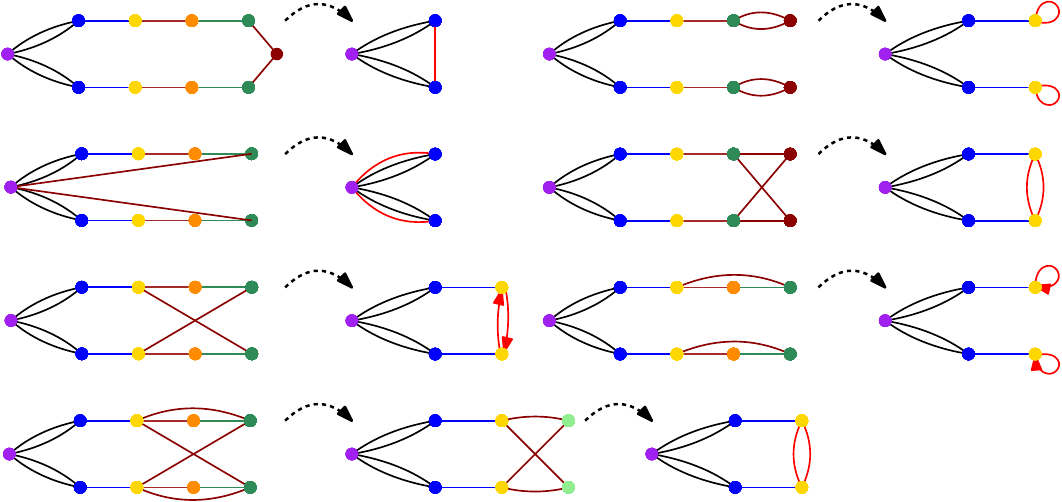}
\caption{Illustration to the reductions of block graphs containing $FW(2)$.}
\label{fig:FW(2)-caseanalysisIII}
\end{figure}

The next step is to show that all block graphs from Figure~\ref{fig:FW(2)-caseanalysis} define NP-complete covering problems. There is a large number of them, but fortunately they can be sorted into a few groups each of which can be handled en bloc. The following lemma is very useful in this direction, since it shows that we do not need to deal with the primed graphs separately.

\begin{lemma}\label{lem:depriming}
Let $H'$ be a graph with two equivalence classes $A,B$ of the degree partition such that $|A|=|B|$ and the induced \interbg\ $H'[A\cup B]$ is a $k$-fold matching for some $k>0$. Let $H$ be the graph obtained from $H'$ by contracting the matching edges. Then {\sc $H$-Cover} polynomially reduces to {\sc $H'$-Cover}, also when both problems are considered only for simple input graphs. 
\end{lemma}

A $k$-tuple matching is obtained from a matching by replacing every edge by $k$ parallel edges. A contraction of a $k$-tuple matching results in identifying each pair of matched vertices and deleting the (formerly) matching edges. The contracted vertices are recoloured by a new extra colour to ensure that they still form a block of the degree partition. The colours of all other vertices, and of all the edges, are kept unchanged. This means that for every edge of $H$ incident with one of the contracted vertices, it is uniquely determined (by its colour) whether it is inherited from an edge incident with a vertex from $A$, or from $B$. Note that Lemma~\ref{lem:depriming} is stated (and proven) in a general setting when $|A|=|B|$ may be arbitrarily large. However, we will only use it for $|A|=|B|=2$ in the sequel.

\begin{proof}
Denote the vertices of $A$ and $B$ as $A=\{a_1,a_2,\ldots,a_t\}$ and $B=\{b_1,b_2,$ $\ldots, b_t\}$ so that $a_i$ is matched to $b_i$, $i=1,2,\ldots,t$. Let $c_i$ be the vertex obtained by contracting the multiple edge joining $a_i$ to $b_i$, for every $i=1,2,\ldots,t$, and let $C=\{c_1,\ldots,c_t\}$ be the block of $H$ containing the contracted vertices. 

Let a simple graph $G$ be the input of {\sc $H$-Cover} and let $V_C$ be the block of $G$ corresponding to the block $C$ of $H$. Replace every vertex $w\in V_C$ by two vertices $w_a,w_b$, and write $\widetilde{V}_A=\{w_a:w\in V_C\}$ and $\widetilde{V}_B=\{w_b:w\in V_C\}$. For every edge $uw$, retain the edge $uw_a$ if the colour of $uw$ is a colour used for edges incident with vertices of $A$ in $H$, and retain the edge $uw_b$ if that colour is a colour used for edges incident with vertices of $B$. (If both $u$ and $w$ are in $V_C$, we retain the edge $u_aw_a$ or $u_bw_b$, depending on whether the colour of the edge $uw$ is a colour of edges within $A$ or within $B$ in $H'$.)   Denote by $\widetilde{G}$ the graph constructed by this process. Then take $k$ disjoint copies of $\widetilde{G}$, with the copy of $u\in V(\widetilde{G})$ in the $i$-th copy being denoted by $u^i$, for $i=1,2,\ldots,k$. Finally, add the edges of a complete bipartite graph with bipartition classes $\{w^i_a:i=1,2,\ldots,k\}$ and $\{w^i_b:i=1,2,\ldots,k\}$ for every $w\in V_C$, and colour these edges with the same colour as the $k$-fold matching in $H'$ connecting the vertices of $A$ to the vertices of $B$. Denote the resulting graph by $G'$. (Note that $\bigcup_{i=1}^k\widetilde{V}^i_A$ is a block of $G'$ corresponding to $A$ and    $\bigcup_{i=1}^k\widetilde{V}^i_B$ is a block of $G'$ corresponding to $B$.)
We claim that $G'$ covers $H'$ if and only $G$ covers $H$.

Suppose first that $G'$ covers $H'$, and let  $f':G'\to H'$ be a covering projection. Define a mapping $f:G\to H$ as follows. Pay attention only to the first copy of $\widetilde{G}$ and forget about the others. Define the vertex mapping part of $f$ by
$$
f(w)=\left\lbrace
\begin{array}{ll}
f'(w) & \mbox{ if } w\not\in V_C,\\
c_i \mbox{ such that } f'(w^1_a)=a_i \mbox{ and }f'(w^1_b)=b_i & \mbox{ if } w\in V_C
\end{array}
\right.
$$     
and define the edge mapping part of $f$ accordingly. Observe the following. If $w\in V_C$, we have $f'(w^1_a) =a_i$ and $f'(w^1_b)=b_j$ for some $i$ and $j$. Since the edge $w^1_aw^1_b$ is mapped onto a matching edge of $H'$, we conclude that $i=j$, and hence $f$ is correctly defined.

If, on the other hand, $f:G\to H$ is a covering projection, we set (for every $j=1,2,\ldots,k$)
$$
f'(u^j)=\left\lbrace
\begin{array}{ll}
f(u) & \mbox{ if } u\not\in V_C,\\
a_i \mbox{ such that } f(w)=c_i & \mbox{ if } u=w_a\in \widetilde{V}_A,\\
b_i \mbox{ such that } f(w)=c_i & \mbox{ if } u=w_b\in \widetilde{V}_B
\end{array}
\right.
$$
and define the edge mapping part of $f'$ on the copies of $\widetilde{G}$ accordingly. It remains to define the mapping $f'$ on the edges of the added complete bipartite graphs. For every $w\in V_C$, partition the edges of the complete bipartite graph on $w^1_a,w^2_a,\ldots,w^k_a$, $w^1_b,w^2_b,\ldots,w^k_b$ into $k$ perfect matchings and map the edges of $j$-th perfect matching onto the $j$-th edge connecting $a_i$ to $b_i$ in $H'$, for the $i$ such that $f(w)=c_i$.  It is easy to see that $f'$ is a covering projection from $G'$ onto $H'$.
\end{proof}

Note that if $k=1$ in the assumption of Lemma~\ref{lem:depriming}, then actually {\sc $H$-Cover} and {\sc $H'$-Cover} are polynomially equivalent. However, for $k>1$ this is not necessarily the case. E.g., the graph $H_2$ from Figure~\ref{fig:FW(2)-caseanalysis} defines an NP-complete {\sc $H_2$-Cover} problem (as we will see soon), but contracting the red double matching we get $FW(2)$ and {\sc $FW(2)$-Cover} is polynomial-time solvable. 

The following lemma is tailored to graphs whose all blocks of degree partition have sizes at most 2.

\begin{lemma}\label{lem:swap}
Let the blocks of the degree partition of a graph $H$ be $W_1,W_2,\ldots,W_t$ and let $|W_i|\le 2$ for all $i=1,2,\ldots,t$. Denote the vertices of $H$ by $r_i,g_i\in W_i$ if $|W_i|=2$, and by $p_i\in W_i$ if $|W_i|=1$.  Let $f:G\to H$ be a covering projection from a graph $G$ whose blocks of degree partition are $V_1,V_2,\ldots,V_t$ (in the canonical ordering corresponding to $W_1,\ldots,W_t$). Suppose further that either
\begin{enumerate}
\item $G$ is bipartite, or
\item for every doublet block  of $H$ and every edge-colour, both vertices of this block  are incident with the same number of semi-edges of this colour.
\end{enumerate} Then the vertex mapping $f':V(G)\to V(H)$  defined by
$$
f'(u)=\left\lbrace
\begin{array}{ll}
p_i & \mbox{ if $u\in V_i$ and $|W_i|=1$},\\
g_i & \mbox{ if $u\in V_i$, $|W_i|=2$ and $f(u)=r_i$},\\
r_i & \mbox{ if $u\in V_i$, $|W_i|=2$ and $f(u)=g_i$}
\end{array}
\right.
$$
extends to a covering projection from $G$ to $H$.
\end{lemma}

\begin{proof}
First observe that the vertex mapping $f'$ is degree-obedient. In order to show this, suppose that $u\in V_i$ is a vertex of $G$ incident to some edges of colour $\alpha$. This means that in $H$, edges of colour $\alpha$ are incident with vertices of $W_i$. Let $W_j$ be the other block containing vertices incident with edges of colour $\alpha$ (or $j=i$ if the edges of colour $\alpha$ lie within the block $W_i$).

If $|W_i|=|W_j|=1$, we have $f(u)=f'(u)$ and also $f'(v)=f(v)$ for every vertex $v\in V(G)$ adjacent to $u$ via an edge of colour $\alpha$. Hence $\mbox{deg}^{\alpha}_Gu = \mbox{deg}^{\alpha}_H f(u)= \mbox{deg}^{\alpha}_H f'(u)$ if $\alpha$ is an undirected colour (and likewise for a directed one). 

If $|W_i|=1$ and $|W_j|=2$, the blocks $W_i$ and $W_j$ are different. Then the interblock graph $H[W_i\cup W_j]^{\alpha}$ is isomorphic to $FW(k)$ for some $k$, and $u$ is incident with $k$ edges of colour $\alpha$ that lead to vertices mapped onto $r_j$ and to $k$ edges that lead to vertices mapped $g_j$ by $f$. The mapping of these neighbours is swapped in $f'$, but the counts remain the same. Similarly, if $|W_i|=2$ and $|W_j|=1$, $H[W_i\cup W_j]^{\alpha}\simeq FW(k)$ for some $k$, and every vertex of $G$ that is mapped to $r_i$ ($g_i$, respectively) by $f$ is incident with $k$ edges of colour $\alpha$ that lead to vertices mapped onto $p_j$ by $f$ (and also by $f'$). Thus every vertex mapped to $g_i$ ($r_i$, respectively) by $f'$    is incident with $k$ edges of colour $\alpha$ that lead to vertices mapped onto $p_j$ by $f'$.

If $|W_i|=|W_j|=2$ and $i\neq j$, then $H[W_i\cup W_j]^{\alpha}\simeq WW(b,c)$ for some $b,c$. Every vertex mapped by $f$ to $r_i$ ($g_i$, respectively) is incident with $b$ edges leading to vertices mapped onto $r_j$ ($g_j$, respectively) and with $c$ edges leading to vertices mapped onto $g_j$ ($r_j$, respectively) by $f$. With respect to $f'$, this means that every vertex mapped  to $g_i$ ($r_i$, respectively) is incident with $b$ edges leading to vertices mapped onto $g_j$ ($r_j$, respectively) and with $c$ edges leading to vertices mapped onto $r_j$ ($g_j$, respectively).

If $|W_i|=|W_j|=2$ and, moreover, $W_i=W_j$, we have to consider two cases --- $H[W_i]^{\alpha}\simeq W(k,m,\ell,p,q)$ for some $k,m,\ell,p,q$ such that $k+2m=2p+q$ and $H[W_i]^{\alpha}\simeq WD(c,b,c)$ for some $b,c$. 

In the former case, every vertex which is mapped onto $r_i$ ($g_i$, respectively) by $f$ is incident with $k+2m$ edges that lead to vertices mapped onto $r_i$ ($g_i$, respectively) and it is incident with $\ell$ edges that lead to vertices mapped onto $g_i$ ($r_i$, respectively). With respect to $f'$, this means that every vertex mapped  to $g_i$ ($r_i$, respectively) is incident with $2p+q=k+2m$ edges leading to vertices mapped onto $g_j$ ($r_j$, respectively) and with $\ell$ edges leading to vertices mapped onto $r_j$ ($g_j$, respectively), by $f'$.

In the latter case, every vertex which is mapped onto $r_i$ ($g_i$, respectively) by $f$ is incident with $c$ outgoing edges and $c$ incoming ones that lead to vertices mapped onto $r_i$ ($g_i$, respectively) and it is incident with $b$ outgoing edges and $b$ incoming ones that lead to vertices mapped onto $g_i$ ($r_i$, respectively). With respect to $f'$, this means that every vertex mapped  to $g_i$ ($r_i$, respectively) is incident with $c$ outgoing and $c$ incoming edges leading to vertices mapped onto $g_j$ ($r_j$, respectively) and with $b$ outgoing and $b$ incoming edges leading to vertices mapped onto $r_j$ ($g_j$, respectively), by $f'$. 

\smallskip\noindent
{\em Case 1.} If $G$ is bipartite, every degree-obedient vertex mapping can be extended to a covering projection~\cite{n:BFHJK21-MFCS}.

\smallskip\noindent
{\em Case 2.} The condition that the  vertices of each doublet block are incident with the same number of semi-edges of any colour implies that the mapping $id'$ from $H$ to $H$ is an automorphism of $H$ (here $id$ stands for the identity mapping). Furthermore, together with the edge mapping from $\overline{E}(G)\cup\overline{L}(G)\cup S(G)$ to $\overline{E}(H)\cup\overline{L}(H)\cup S(H)$ defined by the following rule, it is a covering projection from $G$ to $H$ that extends $f'$:  
\begin{itemize}
\item If an edge of colour $\alpha$ connecting $u$ to $v$ is mapped onto the $\kappa$'s normal edge  of colour $\alpha$ connecting $f(u)$ and $f(v)$ (or $\kappa$'s loop or $\kappa$'s  semi-edge of colour $\alpha$ incident with $f(u)=f(v)$), map this edge onto the $\kappa$'s edge of colour $\alpha$ connecting $f'(u)$ and $f'(v)$ (or $\kappa$'s loop or  $\kappa$'s semi-edge of colour $\alpha$ incident with $f'(u)=f'(v)$, respectively). \qedhere
\end{itemize}
\end{proof}

\begin{proposition}\label{prop:reducedFW(2)NPc}
Let $H$ be any of the following graphs (see Figure~\ref{fig:FW(2)-caseanalysis})
\begin{itemize}
  \item $A_1, A_2, A'_2, B_k, k\ge 1$,
  \item $B'_k, k\ge 2$,
  \item $C_k, k\ge 0$,
  \item $C'_k, D_k,D'_k, k\ge 1$,
  \item $E_{(2,2)}, E'_{(2,2)}$, $E_{(2,1)}, E'_{(2,1)}$, $F,F'$, $H_k,H'_k,k\ge 3$,
  \item $J,J'$, $L_k, L'_k, M_k,M'_k,k\ge 1$,
  \item $N,N'$.
\end{itemize}
Then the {\sc $H$-Cover} problem is NP-complete for simple input graphs. 
\end{proposition}

\begin{proof}
In all the cases, we reduce from {\sc 2-in-4-Sat} for all positive formulas. Let the vertices of $W_2$ be denoted by $r=r_2$ and $g=g_2$, and the vertex of $W_1$ by $p=p_1$. Vertex $p$ is connected by two parallel edges of colour $\alpha$ to $r$ and by two parallel edges of colour $\alpha$ to $g$. Given a formula $\Phi$ with a set $C$ of clauses over a set of variables $X$, we introduce a vertex $z(s)$ for every clause $s\in C$ and colour it so that $z(s)\in V_1$ (that means that in any covering projection $f:G_{\Phi}\to H$, $f(z(s))=p$). For every pair $x\in X, s\in C$ such that  $x\in s$, we introduce a connector vertex $w(x,s)$ from $V_2$, and connect it by an edge of colour $\alpha$ to $z(s)$. The connector vertices of each variable $x$ are interconnected by a variable gadget that ensures that $f(w(x,s))=f(w(x,s'))$ for any two clauses $s$ and $s'$ containing $x$. This allows to define a truth valuation $\varphi:X\to \{{\sf true, false}\}$ so that $\varphi(x)={\sf true}$ if $f(w(x,s))=r$ and $\varphi(x)={\sf false}$  if $f(w(x,s))=g$ (for some, i.e., for all clauses $s$ such that $x\in s$). If $f$ is a covering projection, every vertex $z(s)$ has 2 neighbours (along edges of colour $\alpha$) that map onto $r$ and 2 neighbours that map onto $g$, which means that every clause contains 2 variables evaluated to {\sf true} and 2 variables evaluated to {\sf false} by $\varphi$. The opposite direction, constructing a covering projection from a 2-in-4-satisfying truth valuation is more or less straightforward. 

The variable gadgets differ for particular graphs. Also in many cases we need to take several copies and link them together as sort of a garbage collection to be able to construct the desired covering projection from a truth valuation. The following claim captures the common features of many of the cases. Recall that a colour-preserving mapping $f:G\to H$ is a {\em partial covering projection} if the edge mapping is compatible with the vertex mapping, and for every vertex $u$ of $G$, the edges incident with it are mapped injectively into the set of edges of $H$ incident with $f(u)$ (which implies that if $\mbox{deg}_Gu=\mbox{deg}_Hf(u)$, in which case we say that $u$ has {\em full degree},  the edges incident with $u$ in $G$ are mapped bijectively onto the edges incident with $f(u)$ in $H$).

  \medskip\noindent
  {\sffamily\normalsize\bfseries Claim A.} 
  Suppose all blocks of $H$ have at most two vertices and let the blocks and the vertices of $H$ be denoted as in Lemma~\ref{lem:swap}. Moreover, let the vertex $p$ be incident only to the 4 edges of colour $\alpha$ leading to vertices $r$ and $g$ and nothing else. Suppose that either 
there exists a simple graph $G(a_1,a_2,\ldots,a_k)$ with blocks $V_1,V_2,\ldots,V_t$ and specified vertices $a_1,a_2,\ldots,a_k\in V_2$ so that either:
\begin{enumerate}
\item $k\ge 3$,
\item all vertices of $V(G)\setminus\{a_1,\ldots,a_k\}$ have full degree,
\item each of the vertices $a_1,a_2,\ldots,a_k$ is missing exactly one edge of colour $\alpha$, the degrees in all other colours are full,
\item in every partial covering projection $f:G(a_1,a_2,\ldots,a_k)\to H$, $f(a_1)=f(a_2)=\ldots = f(a_k)$, 
\item there exists a partial covering projection $f:G(a_1,a_2,\ldots,a_k)\to H$, $f(a_1)=f(a_2)=\ldots = f(a_k)=r$, and 
\item there exists a partial covering projection $f:G(a_1,a_2,\ldots,a_k)\to H$, $f(a_1)=f(a_2)=\ldots = f(a_k)=g$,
\end{enumerate} 

or there exists a simple graph $G(a_1,a_2,\ldots,a_k;b_1,b_2,\ldots, b_k)$ with blocks $V_1,V_2,\ldots,V_t$ and specified vertices $a_1,a_2,\ldots,a_k,b_1,b_2,\ldots,b_k\in V_2$ so that 
\begin{enumerate}
\item $k\ge 3$,
\item all vertices of $V(G)\setminus\{a_1,\ldots,a_k,b_1,\ldots,b_k\}$ have full degree,
\item each of the vertices $a_1,a_2,\ldots,a_k,b_1,b_2,\ldots,b_k$ is missing exactly one edge of colour $\alpha$, the degrees in all other colours are full,
\item in every partial covering projection $f:G(a_1,a_2,\ldots,a_k)\to H$, $f(a_1)=f(a_2)=\ldots = f(a_k)$, 
\item there exists a partial covering projection $f:G(a_1,a_2,\ldots,a_k;b_1,b_2,\ldots,b_k)\to H$, $f(a_1)=f(a_2)=\ldots = f(a_k)=r$ and $f(b_1)=f(b_2)=\ldots = f(b_k)=g$, and 
\item there exists a partial covering projection $f:G(a_1,a_2,\ldots,a_k;b_1,b_2,\ldots,b_k)\to H$, $f(a_1)=f(a_2)=\ldots = f(a_k)=g$ and $f(b_1)=f(b_2)=\ldots = f(b_k)=r$.
\end{enumerate}
Then the {\sc $H$-Cover} problem is NP-complete for simple input graphs.

\medskip\noindent
{\sffamily\normalsize\bfseries Proof of the Claim.} Consider a formula $\Phi$ as an input of {\sc 2-in-4-SAT} such that every variable occurs in exactly $k$ clauses (and all these occurrences are positive). Deciding if there is a truth valuation such that every clause contains exactly 2 variables evaluated to {\sf true} is NP-complete~\cite{kratochvil2003complexity}. If there exists a graph $G(a_1,\ldots,a_k)$ as described above, take a copy of it for every variable $x$, and denote this copy by $G^x$. Order the clauses that contain $x$ arbitrarily as $s_1,s_2,\ldots,s_k$ and identify the vertices $a_i$ and $w(x,s_i)$ for $i=1,2,\ldots,k$. The union of $G^x, x\in X$ together with the clause vertices $z(s), s\in C$, results in the graph $G_{\Phi}$. All vertices of $G_{\Phi}$ have full degree (each vertex $w(x,s)$ gets its degree completed to fullness by the edge of colour $\alpha$ leading to $z(s)$, and each vertex $z(s)$ is incident to 4 edges of colour $\alpha$ because the clause $s$ contains 4 variables). Suppose $G_{\Phi}$ covers $H$, and consider a covering projection $f$. By the assumption on $G(a_1,\ldots,a_k)$, for every variable $x$, all $w(x,s)$ such that $x\in s$ are mapped onto the same vertex of $W_2$. Set $\varphi(x)={\sf true}$ if $f(w(x,s))=r$ for some (and hence for all) $s$ that contain $x$, and set $\varphi(x)={\sf false}$ otherwise. Since $H$ has 2 parallel edges of colour $\alpha$ leading to vertex $r$ and 2 leading to $g$, for every clause $s\in C$, vertex $z(s)$ is incident to 2 edges of colour $\alpha$ leading to vertices $w(x,s)$ such that $f(w(x,s))=r$ and 2 edges leading to vertices $w(x,s)$ mapped onto $g$. Alas, $s$ contains 2 variables evaluated to {\sf true} and 2 variables evaluated to {\sf false}. On the other hand, let $\varphi:X\to\{{\sf true, false}\}$ be 2-in-4-satisfying truth valuation of $\Phi$. Define a covering projection $f$ as follows. For every variable $x$ such that $\varphi(x)={\sf true}$, take the partial covering projection from $G^x$ onto $H$ which maps all the vertices $w(x,s)$ onto $r$, and for  every variable $x$ such that $\varphi(x)={\sf false}$, take the partial covering projection from $G^x$ onto $H$ which maps all the vertices $w(x,s)$ onto $g$. Of course, $f(z(s))=p$ for all $s\in C$. It remains to define $f$ for the edges $z(s)w(x,s), x\in s, s\in C$. This may require redefining the covering projection on edges of colour $\alpha$ inside the variable gadgets, but that can be done. The vertices that map onto $p$ or $r$ together with the edges of colour $\alpha$ form a 2-regular bipartite graph whose edges can thus be partitioned into 2 disjoint perfect matching. Map the edges of one of the matchings onto one edge of colour $\alpha$ joining $p$ and $r$, and map the edges of the other matching onto the other edge. 

If there exists a graph $G(a_1,\ldots,a_k;b_1,\ldots,b_k)$ as described above, we proceed analogously. We take a copy of this auxiliary graph for every variable $x$ and denote it by $G^x$. Now for every clause $s\in C$, we take 2 vertices $z^1(s)$ and $z^2(s)$, and for every $x\in s$, we introduce $w^1(x,s)$ and $w^2(x,s)$ and make $w^i(x,s)$ adjacent to $z^i(s)$ via an edge of colour $\alpha$, for $i=1,2$. If $s_1,s_2,\ldots,s_k$ are the clauses containing $x$, we identify $w^1(x,s_i)$ with $a_i$, and  $w^2(x,s_i)$ with $b_i$ of $G^x$. Arguments similar to those above show that the graph $G_{\Phi}$ that we have constructed this way, covers $H$ if and only if $\Phi$ is 2-in-4-satisfiable.

\begin{figure}
\centering
\includegraphics[width=0.7\textwidth]{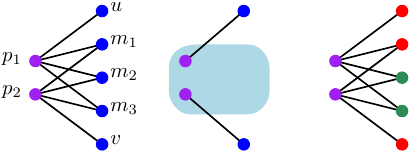}
\caption{The Limping Tripod.}
\label{fig:limpingtripod}
\end{figure}

Finally, we introduce the {\em limping tripod} $LT$ as the graph obtained from $K_{2,3}$ by adding 2 pendant vertices of degree 1 adjacent to the vertices of degree 3, see Figure~\ref{fig:limpingtripod}. The vertices $p_1,p_2$ are from $V_1$, the vertices $u,v,m_1,m_2,m_3$ from $V_2$ and all edges are of colour $\alpha$. A simple but crucial observation is that in every partial covering projection $f$ from $LT$ onto $H$, the vertices $p_1$ and $p_2$ have 3 common neighbours, and since both of them have 2 neighbours mapped onto $r$ and 2 neighbours mapped onto $g$, necessarily $f(u)=f(v)$. This is depicted in Figure~\ref{fig:limpingtripod} right, where we use the colours red and green for indicating which of the vertices of $W_2$ a vertex of the source graph under consideration is mapped onto.
In the middle of the figure we introduce a pictogram for $LT$ that we involve in the more complicated constructions. When we build the variable gadget from several copies of the limping tripod, they are denoted by $LT^1, LT^2, \ldots$ and their vertices are denoted by $u^1,v^1,u^2,v^2,\ldots$ accordingly.
   
   \medskip\noindent
{\em Case $C_0$.} Take 2 copies $LT^1$ and $LT^2$ of the limping tripod and connect $u^1$ to $u^2$ by an edge of colour $\beta$, as well as $v^1$ to $v^2$. In every partial covering projection $f$ to $C_0$, the vertices $u^1$ and $v^1$ are mapped onto the same vertex of $W_2$ because of $LT^1$, while $v^2$ is mapped onto the same vertex as $v^1$ because of the edge $v^1v^2$ of colour $\beta$ which must map onto the semi-edge incident with $f(v^1)$. Hence all vertices $u^1,v^1,u^2,v^2$ must be mapped onto the same vertex, and thus they play the role of $a_1,a_2,a_3,a_4$ in $G(a_1,a_2,a_3,a_4)$. To complete the other vertices to full degrees, add a matching between the $m$-vertices of the two limping tripods. A covering projection that sends all of the $u^1,v^1,u^2,v^2$ vertices onto the red vertex is shown in Figure~\ref{fig:C0-var}, the existence of a    covering projection that maps all of them onto the green vertex follows from Lemma~\ref{lem:swap}. Thus the fact that {\sc $C_0$-Cover} is NP-complete for simple input graphs follows from Claim~A. The NP-completeness of {\sc $C'_0$-Cover} follows from Lemma~\ref{lem:depriming}.

\begin{figure}
\centering
\includegraphics[width=\textwidth]{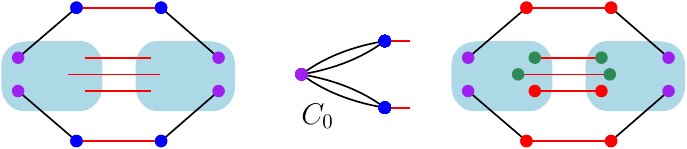}
\caption{The reduction for $C_0$.}
\label{fig:C0-var}
\end{figure}
     
\medskip\noindent
{\em Case $C_k$, $k>0$.} Take $2k+2$ copies of the limping tripod and place a complete graph with edges of colour $\beta$ on the vertices $x^1,x^2,\ldots,x^{2k+2}$, for every $x\in V_2(LT)$. This is $G(u^1,v^1,u^2,v^2,\ldots,u^{2k+2},v^{2k+2})$, and we reduce from {\sc 2-in-4-SAT} for formulas with $4k+4>2$ occurrences per variable. The arguments are analogous to the previous case, it is crucial that a complete graph with even number of vertices is 1-factorizable, i.e., $K_{2k+2}$ covers $F(2k+1,0)$ and hence also $F(1,k)$. Thus {\sc $C_k$-Cover} is NP-complete for simple input graphs for every fixed $k>0$, and the NP-completeness for $C'_k$ follows from Lemma~\ref{lem:depriming}.

\medskip\noindent
{\em Case $D_k, k>0$.} Uses the same idea as the previous case. 
Take $2k+2$ copies of the limping tripod and add edges of colour $\beta$ forming $2k$-regular graphs on $x^1,x^2,\ldots,x^{2k+2}$ for each $x\in V_2(LT)$ to create $G(u^1,v^1,u^2,v^2,\ldots,u^{2k+2},v^{2k+2})$. Reduce from {\sc 2-in-4-SAT} for formulas with every variable appearing in $4k+4\ge 8>2$ clauses, using the fact that the edges of every $2k$-regular graph can be partitioned into $k$ 2-factors, each of which is mapped onto one loop incident with $r$ or $g$. The NP-hardness of {\sc $D'_k$-Cover} then follows from Lemma~\ref{lem:depriming}.

\medskip\noindent
{\em Cases $E_{(2,1)}$ and $E_{(2,2)}$.} Use the construction described in the case of $D_1$, just make sure the 2-regular graph  of colour $\beta$ on  $x^1,x^2,\ldots,x^{2k+2}$ is a Hamiltonian (and hence even) cycle for each $x\in V_2(LT)$. Then $G(u^1,v^1,\ldots,u^4,v^4)$ partially covers $D_1$ if and only if it partially covers $E_{(2,1)}$, which happens if and only if it partially covers $E_{(2,2)}$. Thus both {\sc $E_{(2,1)}$-Cover} and {\sc $E_{(2,2)}$-Cover} are NP-complete, and the NP-completeness of  {\sc $E'_{(2,1)}$-Cover} and {\sc $E'_{(2,2)}$-Cover} follows from Lemma~\ref{lem:depriming}.

\medskip\noindent
{\em Case $H_k$, $k>2$.} Suppose $W_3=\{r',g'\}$ so that $r$ is adjacent to $r'$ and $g$ to $g'$ in $H_k$. Take $k$ copies of the limping tripod, add 
$k$ vertices $x'^1,\ldots,x'^k\in V_3$ and edges of colour $\beta$ forming a connected $k$-regular bipartite graph with classes of bipartition $\{x^1,x^2,\ldots,x^k\}$ and $\{x'^1,\ldots,x'^k\}$ for each $x\in V_2(LT)$. This $G(u^1,v^1,\ldots,u^k,v^k)$ satisfies the assumptions of Claim~A, because the connectedness of the bipartite graph with edges of colour $\beta$ on $u^1,\ldots,u^k,u'^1,\ldots,u'^k$ ensures that $f(u^1)=\ldots=f(u^k)=f(v^1)=\ldots=f(v^k)$ for every partial covering projection $f:G(u^1,v^1,\ldots,u^k,v^k)\to H_k$. For the existence of a partial covering projection, after $f$ is defined on the vertices and edges of the limping tripod part of $G(u^1,v^1,\ldots,u^k,v^k)$, set $f(x'^j)=r'$ if  $f(x^j)=r$ (and set $f(x'^j)=g'$ otherwise) for all $x\in V_2(LT)$ and $j=1,2,\ldots,k$ and for the edge mapping on the edges of colour $\beta$ use the fact that the edges of a $k$-regular bipartite graph can be partitioned into $k$ perfect matchings. Hence the NP-completeness of {\sc $H_k$-Cover} follows from Claim~A, and the NP-completeness of {\sc $H'_k$-Cover} from Lemma~\ref{lem:depriming}.

\medskip\noindent
{Case $L_k$, $k>0$.} Take $2k+2$ copies of the limping tripod and for each $x\in V_2(LT)$, add directed edges of colour $\beta$ on $x^1,x^2,\ldots,x^{2k+2}$ to form a $k$-in-$k$-out-regular connected graph which can be partitioned into $k$ collections of disjoint oriented cycles (this can be achieved, e.g., by taking a $2k$-regular undirected graph, partitioning its edge set into $2k$ perfect matchings, taking the matchings in pairs and orienting the edges so that the cycles formed by the edges of such a pair are directed cycles). Then the graph $G(u^1,\ldots,u^{2k+2}, v^1,\ldots, v^{2k+2})$ constructed in this way satisfies the assumptions of  Claim~A and the NP-completeness of {\sc $L_k$-Cover} follows, as well as the NP-completeness of {\sc $L'_k$-Cover} (using Lemma~\ref{lem:depriming}).

\begin{figure}
\centering
\includegraphics[width=\textwidth]{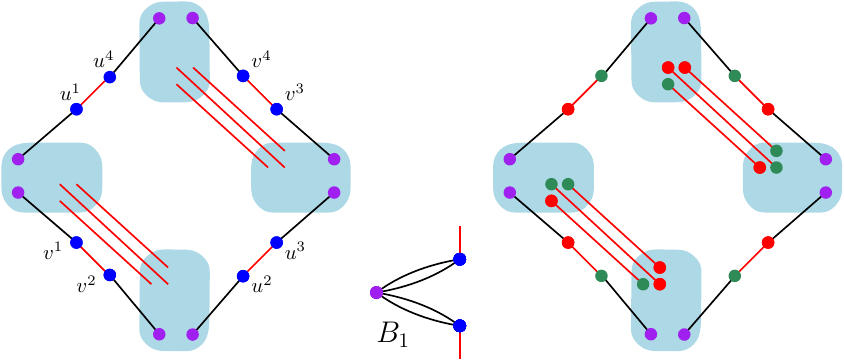}
\caption{The reduction for $B_1$.}
\label{fig:B1-var}
\end{figure}

\medskip\noindent
{Case $B_1$.} Take 4 copies of the limping tripod and add a matching in colour $\beta$ containing $x^1x^2$ and $x^3x^4$ for $x=m_1,m_2,m_3,v$ and $u^1u^4$ and $u^2u^3$ (as depicted in Figure~\ref{fig:B1-var}). In every partial covering projection $f$ onto $B_1$, it must be $f(u^1)\neq f(u^4)=f(v^4)\neq f(v^3)=f(u^3)\neq f(u^2)=f(v^2)\neq f(v^1)=f(u^1)$ since $f(u^1u^4)=f(v^4v^3)=f(u^3u^2)=f(v^2v^1)$ is the edge of colour $\beta$ connecting $r$ and $g$ in $B_1$. Hence the constructed graph fulfills the requirements of $G(u^1,u^3,v^1,v^3;u^2,u^4,v^2,v^4)$, the construction of a partial covering projection such that $f(u^1)=f(v^1)=f(u^3)=f(v^3)=r$ and  $f(u^2)=f(v^2)=f(u^4)=f(v^4)=g$ is illustrated in Figure~\ref{fig:B1-var} right, the existence of the swapped covering projection follows from Lemma~\ref{lem:swap}. The NP-completeness of {\sc $B_1$-Cover} then follows from  Claim~A, and for {\sc $B'_1$-Cover} from Lemma~\ref{lem:depriming}.

\medskip\noindent
{Case $B_k$, $k>1$.} 
Take $2k$ copies of the limping tripod and add, for every $x\in V_2(LT)$, edges of colour $\beta$ forming a complete bipartite  graph with classes of bipartition $\{x^1,x^3,\ldots, x^{2k-1}\}$ and $\{x^2,\ldots,x^{2k}\}$. This graph fulfills the properties of $$G(u^1,u^3,\ldots, u^{2k-1},v^1,v^3,\ldots, v^{2k-1}; u^2,u^4,\ldots,u^{2k},v^2,v^4,\ldots,v^{2k})$$ in every partial covering projection $f$ onto $B_k$, the complete bipartite graphs in colour $\beta$ imply that $f(u^{2i-1})\neq f(u^{2j})$ for every $i,j=1,2,\ldots, k$, and hence $f(u^1)=f(u^3)=\ldots f(u^{2k-1})\neq f(u^2)=f(u^4)=\ldots f(u^{2k})$, and $f(v^i)=f(u^i)$, $i=1,\ldots, 2k$ follows from the properties of the limping tripod. A feasible colouring can be designed e.g. as follows -- set $f(x^{2i-1})=r$ and  $f(x^{2i})=g$ for $i=1,2,\ldots,k$ and $x=u,v,m_1$, and set  $f(x^{2i-1})=g$ and  $f(x^{2i})=r$ for $i=1,2,\ldots,k$ and $x=m_2,m_3$. Since $B_k$ has no semi-edges, this vertex mapping can always be extended to a partial covering projection. Thus {\sc $B_k$-Cover} is NP-complete by  Claim~A, and  {\sc $B'_k$-Cover} by Lemma~\ref{lem:depriming}.
    
\medskip\noindent
{Case $M_k$, $k>0$.} Take $4k$ copies of the limping tripod. For every $x\in V_2(LT)$, add edges of colour $\beta$ forming a complete bipartite graph with classes of bipartition $\{x^1,x^3,\ldots,x^{4k-1}\}$ and $\{x^2,x^4,\ldots,x^{4k}\}$ and orient its edges so that it is $k$-in and $k$-out-regular (e.g., by orienting the edge $x^{2i-1}x^{2i+j}$ from $x^{2i-1}$ to $x^{2i+j}$ for $j=1,2,\ldots,k$ and in the opposite direction for $j=k+1,k+2,\ldots,2k$ with addition in superscripts being modulo $4k$). Similarly as in the previous case, this graph fulfills the properties of
$$G(u^1,u^3,\ldots,u^{4k-1},v^1,v^3,\ldots,v^{4k-1};u^2,u^4,\ldots,u^{4k},v^2,v^4,\ldots,v^{4k})$$
and NP-completeness of {\sc $M_k$-Cover} follows from  Claim~A, and then {\sc $M'_k$-Cover} is NP-complete by Lemma~\ref{lem:depriming}. 

\begin{figure}
\centering
\includegraphics[width=\textwidth]{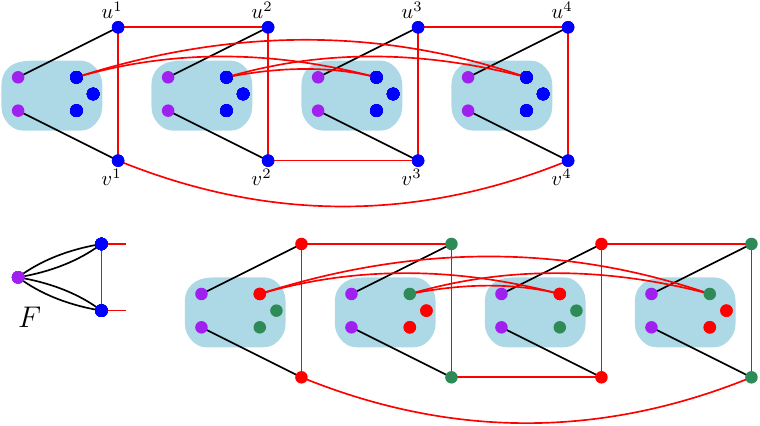}
\caption{The reduction for $F$.}
\label{fig:F-var}
\end{figure}

\medskip\noindent
{Case $F$.} Take 4 copies of the limping tripod and add edges of colour $\beta$ as follows: $u^iv^i, i=1,2,3,4$, $u^1u^2,v^2v^3, u^3u^4, v^1v^4$, and for every $j=1,2,3$, add $m^1_jm^3_j, m^2_jm^3_j, m^2_jm^4_j, m^1_jm^4_j$. The gadget is illustrated in Figure~\ref{fig:F-var} top left, for the sake of legibility only the edges of one of the $m$ levels are shown. If $f$ is a partial covering projection onto $F$, the properties of the limping tripod ensure that $f(u^i)=f(v^i)$ for $i=1,2,3,4$. Thus $u^1$ has one neighbour, namely $v^1$, mapped by $f$ onto the same vertex of $W_2(F)$, and so $f(u^2)\neq f(u^1)$. Analogously, $f(u^3)=f(v^3)\neq f(v^2)= f(u^2)$ and $f(u^4)\neq f(u^3)$. A covering projection that respects this pattern is depicted in Figure~\ref{fig:F-var} bottom right. Hence this graph $G(u^1,v^1,u^3,v^3;u^2,v^2,u^4,v^4)$ fulfills the properties of  Claim~A and the NP-completeness of {\sc $F$-Cover} follows. For the graph $F'$, the NP-completeness of {\sc $F'$-Cover} follows from Lemma~\ref{lem:depriming}. 

\begin{figure}
\centering
\includegraphics[width=\textwidth]{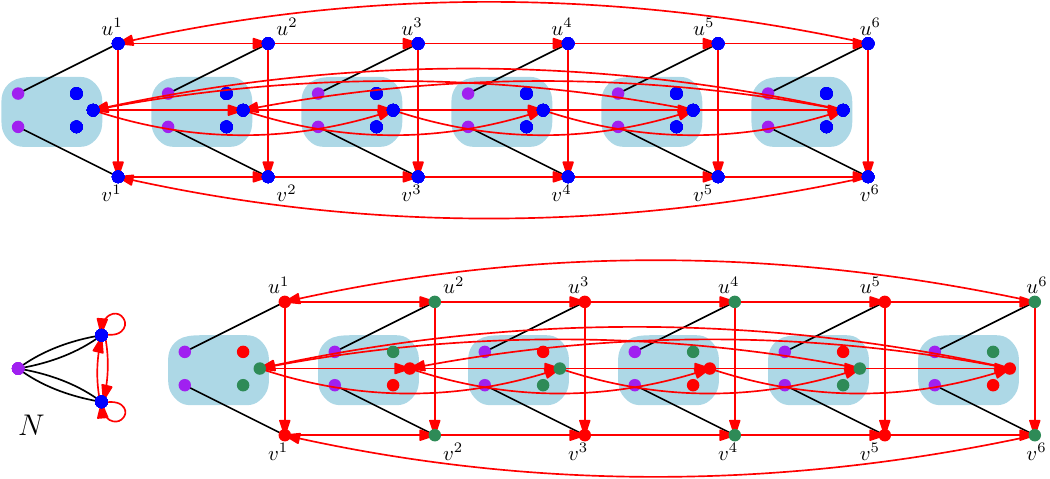}
\caption{The reduction for $N$.}
\label{fig:N-var}
\end{figure}

\medskip\noindent
{Case $N$.} Take 6 copies of the limping tripod and add directed edges of colour $\beta$ as follows: $u^iv^i, i=1,2,\ldots,6$, $u^iu^{i+1}, i=1,2,\ldots,5$, $u^6u^1$, $v^iv^{i+1}, i=1,2,\ldots,5$, $v^6v^1$, and for every $j=1,2,3$, $m^i_jm^{i+1}_j, i=1,2,\ldots,5$, $m^6_j,m^1_j$,  
$m^i_jm^{i+2}_j, i=1,2,3,4,$ $m^5_jm^1_j, m^6_jm^2_j$. The construction is illustrated in Figure~\ref{fig:N-var} top left, where  the edges of colour $\beta$ are illustrated only on one of the $m$-levels. The properties of the limping tripod imply that for every partial covering projection $f$ onto $N$, $f(u^i)=f(v_i)$ for all $i=1,2,\ldots,6$. Since every vertex of $V_2$ has two outgoing edges of colour $\beta$ and their end-vertices  must be mapped onto different vertices of $W_2(N)$, we have $f(u^i)\neq f(u^{i+1})$ for all $i=1,2,\ldots,5$. A partial covering projection respecting this pattern is shown in Figure~\ref{fig:N-var} bottom right. Hence the constructed graph $G(u^1,v^1,u^3,v^3,u^5,v^5;
u^2,v^2,u^4,v^4,u^6,v^6)$ fulfills the properties required by  Claim~A and the NP-completeness of {\sc $N$-Cover} follows. The NP-completeness of {\sc $N'$-Cover} follows by Lemma~\ref{lem:depriming}.

\begin{figure}
\centering
\includegraphics[width=\textwidth]{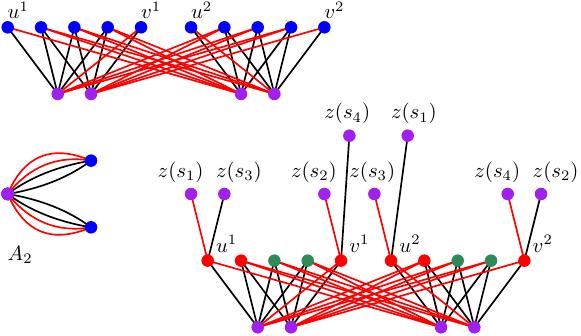}
\caption{The reduction for $A_2$.}
\label{fig:A2-var}
\end{figure}

\medskip\noindent
{Case $A_2$.} We reduce from {\sc 2-in-4-SAT} for formulas with 4 occurrences per variable. Again each clause $s\in C$ is represented by a vertex $z(s)\in V_1(G_{\Phi})$, but now from this vertex there are 2 edges leading to the variable gadget for every variable occurring in $s$, one of colour $\alpha$ and  one of colour $\beta$. The variable gadget is built from 2 copies of the limping tripod and it is depicted in Figure~\ref{fig:A2-var} top left. In colour $\beta$, there are complete bipartite graphs connecting $p^1_1,p^1_2$ to $m^2_1,m^2_2,m^2_3$ and  $p^2_1,p^2_2$ to $m^1_1,m^1_2,m^1_3$, and furthermore the edges $p^1_1v^1, p^1_2v^2,p^2_1u^1,p^2_2u^2$. Thus the vertices $u^1, u^2, v^1, v^2$ are each missing one edge of colour $\alpha$ and one edge of colour $\beta$, the degrees of all other vertices are full (in both colours). In order to keep the graph $G_{\Phi}$ simple, the connector edges of colours $\alpha$ and $\beta$ from the variable gadgets to the clause gadgets are shifted as follows: if a variable $x$ occurs in clauses $s_1.s_2,s_3$ and $s_4$, then vertex $u^1$ of the variable gadget of $x$ is adjacent to $z(s_1)$ via an edge of colour $\alpha$ and to $z(s_3)$ via an edge of colour $\beta$. And similarly, $v^1$ ($u^2$, $v^2$, respectively) is adjacent via an edge of colour $\alpha$ to $z(s_2)$ ($z(s_3), z(s_4)$, respectively) and via an edge of colour $\beta$ to $z(s_4)$ ($z(s_1), z(s_2)$, respectively). If $G_{\Phi}$ covers $A_2$ and $f$ is a covering projection, then the properties of limping tripods imply $f(u^1)=f(v^1)=f(u^2)=f(v^2)$ in every variable gadget, and setting $\varphi(x)={\sf true}$ iff $f(u^1)=r$ in the variable gadget of variable $x$ guarantees that the same information about the truth valuation of $x$ reaches every clause containing it. And in every clause $s\in C$ there are exactly 2 variables evaluated to {\sf true} and exactly 2 clauses evaluated to {\sf false}, $z(s)$ is adjacent to vertices of 4 variable gadgets, out of which 2 must be mapped onto $r$ and 2 onto $g$ by $f$. On the other hand, the variable gadget allows a partial covering projection onto $A_2$ in which all the vertices $u^1,v^1,u^2,v^2$ are mapped onto $r$ (and a companion partial covering projection in which these vertices are mapped onto $g$) as depicted in Figure~\ref{fig:A2-var} bottom right. Thus if $\varphi$ is truth valuation which 2-in-4-satisfies $\Phi$, we can compose a degree-obedient vertex mapping $f:V(G_{\Phi})\to \{p,r,g\}$, and since $G_{\Phi}$ has no semi-edges, this vertex mapping extends to a covering projection onto $A_2$. This concludes the proof of NP-completeness of {\sc $A_2$-Cover}, and the NP-completeness of {\sc $A'_2$-Cover} follows by Lemma~\ref{lem:depriming}.           

\begin{figure}
\centering
\includegraphics[width=\textwidth]{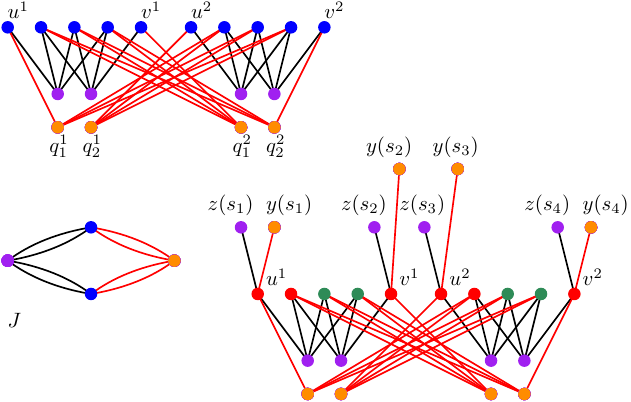}
\caption{The reduction for $J$.}
\label{fig:J-var}
\end{figure}

\medskip\noindent
{Case $J$.} This case is similar to the case of $A_2$. Denote by $q$ the vertex of block $W_3$ of $J$. For every clause $s$ of $\Phi$, $G_{\Phi}$ has 2 vertices, $z(s)\in V_1(G_{\Phi})$ and $y(s)\in V_3(G_{\Phi})$. The variable gadget is based on 2 copies of the limping tripod $LT^1, LT^2$ with edges of colour $\alpha$ which are further connected by   another two copies of limping tripod, these two having edges of colour $\beta$, on vertices $q^1_1,q^1_2,q^2_1,q^2_2\in V_3(G_{\Phi})$ and the $u,v$ and $m$ vertices of $LT^1$ and $LT^2$ as follows: edges of colour $\beta$ form a complete bipartite graph with classes of bipartition $\{q^1_1,q^1_2\}$ and $\{u^2,m^2_1,m^2_2,m^2_3,v^2\}$, another one with classes of bipartition  $\{q^2_1,q^2_2\}$ and $\{u^1,m^1_1,m^1_2,m^1_3,v^1\}$, and edges $q^1_1u^1,q^1_2u^2,q^2_1v^1,q^2_2v^2$. Again, each of the vertices $u^1,u^2,v^1,v^2$ misses one edge of colour $\alpha$ and one edge of colour $\beta$, all other vertices of the variable gadget are of full degree. If a variable $x$ occurs in clauses $s_1,s_2,s_3,s_4$, the variable gadget of $x$ is connected to the clause vertices by edges $u^1z(s_1),v^1z(s_2),u^2z(s_3),v^2z(s_4)$ of colour $\alpha$ and by edges $u^1y(s_1),v^1y(s_2),u^2y(s_3),v^2y(s_4)$ of colour $\beta$. As in the case of $A_2$, if $f:G_{\Phi}\to J$ is a covering projection, then $\varphi:X\to \{{\sf true,false}\}$ defined by $\varphi(x)={\sf true}$ iff $f(u^1)=r$ for the copy of $u^1$ in the variable gadget for $x$, is a 2-in-4-satisfying truth valuation of $\Phi$. On the other hand, if $\Phi$ is 2-in-4-satisfiable, a covering projection of $G_{\phi}$ onto $J$ can be built from partial covering projections depicted in Figure~\ref{fig:J-var} bottom right (the variable gadget is depicted in top left of the same figure).

\begin{figure}
\centering
\includegraphics[width=\textwidth]{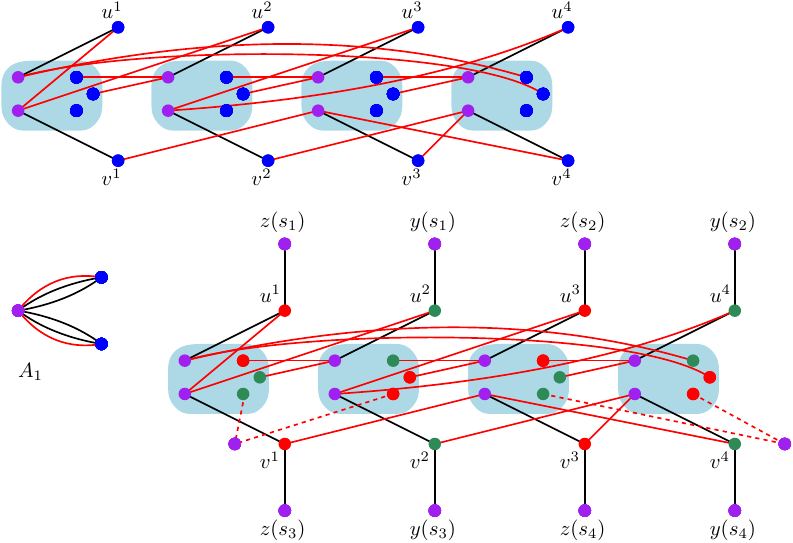}
\caption{The reduction for $A_1$.}
\label{fig:A1-var}
\end{figure}

\medskip\noindent
{Case $A_1$.} The last reduction builds upon the ideas of the previous ones, but is slightly more involved. We reduce from {\sc 2-in-4-SAT} for formulas with 4 occurrences per variable. For every clause $s\in C$, we introduce two vertices $z(s),y(s)\in V_1(G_{\Phi})$. 

The variable gadget, depicted in Figure~\ref{fig:A1-var} top left, consists of 4 copies of the limping tripod with edges of colour $\alpha$. Edges of colour $\beta$ are added as follows:
$$p^1_1m^4_1,p^1_1m^4_2, p^2_1m^1_1,p^2_1m^1_2,
  p^3_1m^2_1,p^3_1m^2_2,
  p^4_1m^3_1,p^4_1m^3_2,
  p^1_2u^1, p^1_2u^2, p^2_2u^3, p^2_2u^4, 
  p^3_2v^1,
  p^3_2v^4, p^4_2v^2, p^4_2v^3.$$
In this graph, vertices $u^i, v^i, i=1,2,3,4$ are missing one edge of colour $\alpha$ each, and vertices $m^i_3, i=1,2,3,4$ are missing one edge of  colour $\beta$ each, all other vertices are of full degree. In every partial covering projection $f$ onto $A_1$, the properties of limping tripods imply that $f(u^i)=f(v^i)$ for $i=1,2,3,4$. Since $f(p^1_2)=p$, the neighbours of $p^1_2$ must be mapped onto different neighbours of $p$ in $A_1$, i.e., one of them is mapped on $r$ and the other one on $g$. Hence $f(u^1)\neq f(u^2)$. Analogously $f(u^2)=f(v_2)\neq f(v_3)=f(u_3)$ (because of $v_2,v_3$ being neighbours of $p^3_2$ along edges of colour $\beta$) and $f(u^3)\neq f(u^4)$ (because they are neighbours of $p^2_2$ along edges of colour $\beta$). Hence $f(u^1)=f(v^1)=f(u^3)=f(v^3)=r$ and $f(u^2)=f(v^2)=f(u^4)=f(v^4)=g$, or vice versa. 

The connections of the variable gadgets to the clause gadgets are described in two steps. Let $x\in X$ be a variable and let $s_1,s_2,s_3,s_4$ be the clauses containing $x$. We connect the clause vertices $z(s_1),z(s_2),z(s_3)$ and $z(s_4)$ to  $u^1,u^3,v^1$ and $v^3$, respectively, and the clause vertices $y(s_1),y(s_2),y(s_3)$ and $y(s_4)$ to  $u^2,u^4,v^2$ and $v^4$, again in this order, by edges of colour $\alpha$. At this point, overall in the graph $G_{\Phi}$ as constructed so far, we have $2|C|$ vertices ($z(s),y(s), s\in C$) that are missing 2 edges of colour $\beta$ each, and $4|X|$ vertices ($m^i_3,i=1,2,3,4$ in the variable gadgets) that are missing one edge of colour $\beta$ each, all other vertices are of full degree. View the later vertices grouped in pairs $\{m^1_3,m^2_3\}, \{m^3_3, m^4_3\}$ within the variable gadgets. Since $|C|=|X|$, the number of such pairs is equal to the number of clause vertices. We make each clause vertex adjacent to both vertices of one pair via edges of colour $\beta$, in an arbitrary one-to-one correspondence of the clause vertices to the pairs. This concludes the construction of $G_{\Phi}$.

If $f:G_{\Phi}\to A_1$ is a covering projection, then $\varphi:X\to \{{\sf true,false}\}$ defined by $\varphi(x)={\sf true}$ iff $f(u^1)=f(u^3)=f(v^1)=f(v^3)=r$ 2-in-4-satisfies $\Phi$, because every vertex $z(s)$ has exactly two $\alpha$-neighbours mapped onto $r$ and exactly two of them mapped onto $g$, and so every clause $s$ has exactly two variables evaluated to {\sf true}.

On the other hand, take the vertex mapping indicated in Figure~\ref{fig:A1-var} bottom right for the variables evaluated to {\sf true}, and the primed mapping (i.e., red and green vertices being swapped) for the variables evaluated to {\sf false}. Then every clause vertex $z(s)$ has two neighbours mapped onto $r$ and two neighbours mapped onto $g$, but so does $y(s)$ as well. The two $\beta$-neighbours  of every $V_1$ vertex are mapped onto different vertices from $W_2(A_1)$ (for the $p^i_j$ vertices within variable gadgets this is clear from the partial covering from the gadget onto $A_1$ depicted in Figure~\ref{fig:A1-var}, for the clause vertices $z(s),y(s)$, this follows from the construction where each of these vertices has been made adjacent to a pair of red-green vertices in this partial covering). Thus, we see that $G_{\Phi}$ covers $A_1$ if and only if $\Phi$ is 2-in-4-satisfiable and the proof of NP-completeness of {\sc $A_1$-Cover} is concluded.   
\end{proof}

\subsection{Garbage collection}\label{subsec:garbage}

In this subsection we show how the puzzle of particular cases clips together and provide the proof of Parts 2 and 3 of Theorem~\ref{thm:main}. The last step is the garbage collection (i.e., completing the construction so that a covering projection from entire input graph onto the whole $H$ exists, whenever the relevant part of the input graph covers the harmful or dangerous block graph of $H$). In the constructions we rely on the following simple observation.

\begin{observation}\label{obs:perfectmatchings}
\begin{enumerate}
\item For every $k,m$ such that $0<k<m, m$ even, there exists a $k$-edge-colourable $k$-regular graph on $m$ vertices.
\item For every $k,m$ such that $0<k\le m$, there exists a $k$-edge-colourable $k$-regular bipartite graph with $2m$ vertices.
\end{enumerate}
\end{observation}

\begin{proof}
\begin{enumerate}
\item It is well known that for an even $m$, the complete graph on $m$ vertices is of Vizing class 1, i.e., $(m-1)$-regular and $(m-1)$-edge-colourable. Take any $k$ of the $m-1$ colours, the union of the edges of these $k$ colours is the desired graph. Note that if $k>1$, the colours may be chosen so that the graph is connected.
\item Analogously, it is well known that every bipartite graph is of Vizing class 1. Hence every $k$-regular bipartite graph is $k$-edge-colourable. There are many ways how to construct $k$-regular bipartite graphs, for the mere existence we may argue as in the previous case -- take an $m$-edge-colouring of the complete bipartite graph $K_{m,m}$, choose $k$ of the colours and the union of the edges of these $k$ colours. Again, if $k>1$, we may further require that the graph is connected.
\end{enumerate}
\end{proof}

The following observation is easy but useful for simplifying the proofs below.

\begin{observation}\label{obs:span}
Let $H'$ be a block graph of a graph $H$ and let $H''$ be the spanning subgraph of $H$ containing all vertices of $H$ and exactly the edges of $H'$. Then {\sc $H'$-Cover} for simple input graphs polynomially reduces to {\sc $H''$-Cover} for simple input graphs.
\end{observation}

\begin{proof}
Let $G'$ be the simple graph for which we seek a covering onto $H'$. Suppose $V(H)=W_1\cup\ldots\cup W_t$ and let $V(H')=W_1\cup\ldots\cup W_h$. Identify the blocks $V_1,\ldots,V_h$ of $G'$ and check if each $|V_i|$ is divisible by $|W_i|$ for $i=1,2,\ldots,h$, and if the ratio $\frac{|V_i|}{|W_i|}$ is the same for all $i$. Reject the graph if it is not; otherwise, denote this ratio by $k$. In the latter case, construct $G''$ by adding $k\cdot(\sum_{i=h+1}^t |W_i|)$ isolated vertices, vertex coloured so that for every $i=h+1,\ldots,t$, $k\cdot |W_i|$ vertices form a block $V_i$. It is clear that $G''$ covers $H''$ if and only if $G'$ covers $H'$.  
\end{proof}

\begin{definition}\label{def:balanced}
Let $H$ be a graph with at most 2 vertices in each equivalence class of its degree partition. Then $H$ is called {\em balanced} if in every doublet block and every edge-colour $\alpha$, the two vertices of this block are incident with the same number of semi-edges of colour $\alpha$. 
\end{definition}

\begin{proposition}\label{prop:balanced}
Let $H$ be a graph with at most 2 vertices in each equivalence class of its degree partition and let $H'$ be a balanced block graph of $H$. Then {\sc $H'$-Cover} for simple input graphs polynomially reduces to {\sc $H$-Cover} for simple input graphs.
\end{proposition}

\begin{proof}
In view of Observation~\ref{obs:span} we may assume that $H'$ is a spanning block graph of $H$. Let $G'$ be a simple graph whose covering of $H'$ is questioned. Denote by $D$ the maximum total degree of a vertex of $H$, and fix an even number $m>D$. Construct a simple graph $G$ by taking $2m$ copies $G'_{ij},i=1,2, j=1,2,\ldots,m$ of $G'$ and adding edges as follows:

\begin{enumerate}
\item The copy of a vertex $x\in V(G')$ in $G'_{ij}$ is denoted by $x_{ij}$.
\item For every $x\in V(G')$, prepare a pool ${\cal A}(x)$ of $D$ disjoint perfect matchings on vertices $x_{1j}, j=1,2,\ldots,m$,  a pool ${\cal B}(x)$ of $D$ disjoint perfect matchings on vertices $x_{2j}, j=1,2,\ldots,m$, and a pool ${\cal C}(x)$ of $D$ disjoint perfect matchings in the complete bipartite graph with classes of bipartition $\{x_{1j}:j=1,2,\ldots,m\}$ and $\{x_{2j}:j=1,2,\ldots,m\}$;
\item for every two distinct $x,y\in V(G')$,  prepare a pool ${\cal D}(x,y)$ of $D$ disjoint perfect matchings in the complete bipartite graph with classes of bipartition $\{x_{1j}:j=1,2,\ldots,m\}$ and $\{y_{1j}:j=1,2,\ldots,m\}$, a pool ${\cal E}(x,y)$ of $D$ disjoint perfect matchings in the complete bipartite graph with classes of bipartition $\{x_{2j}:j=1,2,\ldots,m\}$ and $\{y_{2j}:j=1,2,\ldots,m\}$, and a pool ${\cal F}(x,y)$ of $D$ disjoint perfect matchings in the complete bipartite graph with classes of bipartition $\{x_{1j}:j=1,2,\ldots,m\}$ and $\{y_{2j}:j=1,2,\ldots,m\}$.
\item If $x\in V_{\sigma}$ such that $W_{\sigma}$ is a singleton block of $H$, and $\alpha$ is an edge-colour that appears within $W_{\sigma}$, do as follows

\begin{enumerate}
\item if $p_{\sigma}\in W_{\sigma}$ is incident with $b$ semi-edges and $c$ undirected loops of colour $\alpha$ in $H$, extract $b+2c$ perfect matchings from ${\cal A}(x)$ and put them into $G$ as edges of colour $\alpha$, and do the same with the pool ${\cal B}(x)$,
\item if $p_{\sigma}\in W_{\sigma}$ is incident with $d$ directed loops of colour $\alpha$ in $H$, extract $2d$ perfect matchings from ${\cal C}(x)$, orient their edges so that every vertex $x_{ij}$ has in-degree and out-degree $d$ and put them into $G$ as oriented edges of colour $\alpha$.
\end{enumerate}

\item  If $x\in V_{\sigma}$ such that $W_{\sigma}$ is a doublet block of $H$, and $\alpha$ is an edge-colour that appears within $W_{\sigma}$, do as follows 

\begin{enumerate}
\item if $H[W_{\sigma}]^{\alpha}\simeq W(b,c,\ell,c',b')$, extract $b+2c=b'+2c'$ perfect matchings from ${\cal A}(x)$ and put them into $G$ as edges of colour $\alpha$, do the same with the pool ${\cal B}(x)$, and extract $\ell$ perfect matchings from ${\cal C}(x)$ and put them into $G$ as edges of colour $\alpha$;
\item if $H[W_{\sigma}]^{\alpha}\simeq WD(b,\ell,b)$, extract $2b$ perfect matchings from ${\cal A}(x)$, orient their edges so that every vertex has in-degree and out-degree $b$ and put them into $G$ as directed edges of colour $\alpha$, do the same with the pool ${\cal B}(x)$, and extract $2\ell$ perfect matchings from ${\cal C}(x)$, orient the edges of $\ell$ of them from $x_{1j}$ to $x_{2j'}$ and of the other $\ell$ matchings from $x_{2j}$ to $x_{1j'}$  and put them into $G$ as directed edges of colour $\alpha$.
\end{enumerate}

\item If $W_{\sigma}$ and $W_{\rho}$ are different singleton blocks, pair the vertices in $V_{\sigma}$ and $V_{\rho}$ (they are $k$ of them in each of these blocks of $G'$) and for each such pair $x\in V_{\sigma}$, $y\in V_{\rho}$, do as follows. If $\alpha$ is an edge-colour such that $H[W_{\sigma}\cup W_{\rho}]^{\alpha}\simeq FF(b)$, extract $b$ disjoint perfect matchings from ${\cal D}(x,y)$ and add them as edges of colour $\alpha$ to $G$, and do the same with ${\cal E}(x,y)$.
\item If $W_{\sigma}$ is a singleton block of $H$ and $W_{\rho}$ is one of its doublet blocks, group their vertices into disjoint triples $x\in V_{\sigma}$, $y_1,y_2\in V_{\rho}$ (note that $|V_{\rho}|=2\cdot |V_{\sigma}|$) and do as follows. If
$\alpha$ is an edge-colour such that $H[W_{\sigma}\cup W_{\rho}]^{\alpha}\simeq FW(b)$, extract $b$ disjoint perfect matchings from ${\cal D}(x,y_1)$ and add them as edges of colour $\alpha$ to $G$, and do the same with ${\cal E}(x,y_2)$, ${\cal F}(x,y_1)$ and ${\cal F}(y_2,x)$.
\item  If $W_{\sigma}$ and $W_{\rho}$ are different doublet blocks of $H$, pair the vertices in $V_{\sigma}$ and $V_{\rho}$ (there are $2k$ of them in each of these blocks of $G'$) and for each such pair $x\in V_{\sigma}$, $y\in V_{\rho}$, do as follows. If $\alpha$ is an edge-colour such that $H[W_{\sigma}\cup W_{\rho}]^{\alpha}\simeq WW(b,c)$, extract $b$ disjoint perfect matchings from ${\cal D}(x,y)$ and add them as edges of colour $\alpha$ to $G$, do the same with ${\cal E}(x,y)$, and extract $c$ disjoint perfect matchings from ${\cal F}(x,y)$ and add them as edges of colour $\alpha$ to $G$, and do the same with ${\cal F}(y,x)$.
\end{enumerate}

We claim that $G$ covers $H$ if and only if $G'$ covers $H'$. Moreover, note that it follows from the construction of $G$ that it is a simple graph.

Clearly, if $G$ covers $H$, then each copy $G'_{ij}$ of $G'$ must cover $H'$. The core of the proof is in the opposite implication. 

Suppose $G'$ covers $H'$ and let $f:G'\to H'$ be a covering projection. Since $H'$ is balanced, the companion mapping $f':G'\to H'$ (cf. Lemma~\ref{lem:swap}) is a covering projection as well. Use $f$ on $G'_{1j}, j=1,2,\ldots,m$ and $f'$ on $G'_{2j}, j=1,2,\ldots,m$. A somewhat tedious but straightforward case analysis shows that the vertex part of this compound mapping is degree-obedient with respect to entire $H$, and the way the additional edges of $G$ were constructed from perfect matchings implies that this mapping extends to a graph covering projection of $G$ onto $H$. 
\end{proof}

\medskip\noindent
{\em Proof of Part 2 of Theorem~\ref{thm:main}.} The undirected harmful uniblock graphs are $F(b,c)$ (with $b\ge 2$ and $b+c\ge 3$), $W(k,m,\ell,p,q)$ (with $\ell\ge 1$ and $k+2m+\ell=q+2p+\ell\ge 3$) and $F(b,c)+F(b',c')$ (with $b'\le b, b\ge 2, b+2c=b'+2c'\ge 3$), the directed harmful uniblock graphs are $WD(c,b,c)$ (with $b,c\ge 1, b+c\ge 3)$. For each of them, the covering problem is NP-complete for simple input graphs as proven in Proposition~\ref{prop:NP-block}. The graphs $F(b,c)$ and $WD(b,c,b)$ are balanced, and thus if $H$ contains a harmful block graph of one of these two types, {\sc $H$-Cover} is NP-complete for simple graphs by  Proposition~\ref{prop:balanced}.

Consider $W(k,m,\ell,p,q)$ (with $\ell\ge 1$ and $k+2m+\ell=q+2p+\ell\ge 3$) being a monochromatic uniblock graph of a graph $H$. A bipartite graph $G'$ covers $W(k,m,\ell,p,q)$ if and only if it allows a vertex colouring by two colours such that every vertex has exactly $k+2m=q+2p$ neighbours of its own colour and $\ell$ neighbours of the other colour, i.e., if and only if it covers $W(k+2m,0,\ell,0,q+2p)$. If $f:G'\to W(k,m,\ell,p,q)$ is a covering projection, the vertex mapping $f:V(G')\to \{r,g\}$ is such a $(k+2m,\ell)$-colouring of $G'$, and so is the companion mapping $f':V(G')\to \{r,g\}$ defined by $f'(x)=r$ iff $f(x)=g$, which then extends to a covering projection $f':G'\to W(k,m,\ell,p,q)$. Take a simple bipartite graph $G'$ as an input of {\sc $(k+2m,\ell)$-Color} and construct a graph $G$ exactly as in the proof of Proposition~\ref{prop:balanced}. The argument about the companion vertex mapping implies that $G$ covers $H$ if and only if $G'$ allows a $(k+2m,\ell)$-colouring. And since {\sc $(k+2m,\ell)$-Color} is NP-complete for bipartite input graphs (as proven in \cite{n:BFHJK21-MFCS}), the NP-completeness of {\sc $W(k,m,\ell,p,q)$-Cover} for simple input graphs follows.

Finally, consider $W(b,c,0,c',b')=F(b,c)+F(b',c')$ with vertices $r,g$ (and with parameters $b+2c=b'+2c'\ge 3$, $b\ge 2$, $b'\le b$) being a monochromatic uniblock graph of a graph $H$. The parameters are such that {\sc $F(b,c)$-Cover} is NP-complete for simple graphs, as proven in \cite{n:BFHJK21-MFCS}. Take a connected simple graph   $G''$ as an input of this problem, and let $G'=2G''$ be the union of two disjoint copies of $G''$.    A covering projection onto $F(b,c)$ can be straightforwardly modified to a covering projection onto $F(b',c')$ (if $k=c'-c$, it is $b-b'=2k$ and taking the perfect matchings formed by preimages of the semi-edges of $F(b,c)$ two by two $k$ times, and mapping the union of two perfect matchings onto one loop of $F(b',c')$ creates the covering projection onto $F(b',c')$). Hence $G'$ covers $F(b,c)+F(b',c')$ if and only if $G''$ covers $F(b,c)$. Moreover, if $f:G'\to F(b,c)+F(b',c')$ is a covering projection, then the vertices of copy of $G''$ map onto $r$ and the vertices of the other copy onto $g$, and the companion mapping $f'$ is also a covering projection. Thus taking $G'$ and constructing a simple graph $G$ as in the proof of Proposition~\ref{prop:balanced}, we see that $G$ covers $H$ if and only if $G'$ covers   $F(b,c)+F(b',c')$, which happens if and only if $G''$ covers $F(b,c)$. Therefore {\sc $H$-Cover} is NP-complete for simple input graphs.

The monochromatic interblock graphs are $FW(c)$ (with $c\ge 3$) and $WW(b,c)$ (with $b,c\ge 1, b+c\ge 3$). For each of them, the covering problem is NP-complete for simple input graphs as proven in Proposition~\ref{prop:NP-interblock-easy}, and since both of them are balanced, the {\sc $H$-Cover} problem is NP-complete for any graph $H$ that has a monochromatic interblock graph isomorphic to one of them by Proposition~\ref{def:balanced}.  
\qed

\medskip\noindent
{\em Proof of Part 3 of Theorem~\ref{thm:main}.} If $H$ contains a harmful monochromatic uniblock or interblock graph, then {\sc $H$-Cover} is NP-complete by Part 2 of Theorem~\ref{thm:main}. If $H$ does not contain any harmful block graph, but contains a dangerous one, then $H$ contains a block graph $H'$ which is reducible to one of the graphs $A_1, A_2, A'_2, B_k, k\ge 1, B'_k, k\ge 2, C_k, k\ge 0, C'_k, k\ge 1, D_k, k\ge 1, D'_k, k\ge 1, E_{(2,2)}, E'_{(2,2)}, E_{(2,1)}, E'_{(2,1)}, F, F', H_k, k\ge 3, H'_k, k\ge 3, J, J', L_k, k\ge 1, L'_k, k\ge 1, M_k, k\ge 1, M'_k, k\ge 1, N$ and $N'$ (Proposition~\ref{prop:NP-interblock-uneasy}). For each of these graphs, say $H''$, the NP-completeness of {\sc $H''$-Cover} for simple input graphs is proven in Proposition~\ref{prop:reducedFW(2)NPc}. Therefore for the corresponding block graph $H'$ (which is reducible to the $H''$), the {\sc $H'$-Cover} problem is NP-complete for simple input graphs by the nature of reducibility. In all cases except of $E_{(2,1)}$ and $E'_{(2,1)}$, the graph $H''$, and therefore also $H'$, is balanced, and hence the NP-completeness of {\sc $H$-Cover} for simple input graphs follows from Proposition~\ref{prop:balanced}.

The graphs $E_{(2,1)}$ and  $E'_{(2,1)}$  need an extra care. Consider first $E_{(2,1)}$. In the proof of this case in Proposition~\ref{prop:reducedFW(2)NPc}, a simple graph $G_{\Phi}$ is constructed such that $G_{\Phi}$ covers $E_{(2,1)}$ if and only if it covers $E_{(2,2)}$ which happens if and only if it covers $D_1$. For such an input graph, the companion vertex mapping $f'$ to a covering projection $f$ again extends to a covering projection onto $E_{(2,2)}$, and hence also onto $E_{(2,1)}$. Therefore, the proof of Proposition~\ref{prop:balanced} applies because of the nature of the input graph, despite the fact that $E_{(2,1)}$ is not balanced.
For $E'_{(2,1)}$, the argumentation is analogous.
\qed

\subsection{Proof of Theorem~\ref{thm:newmain}}
Now we have everything ready to prove the main result of the paper.

\medskip\noindent
{\em Proof of Theorem~\ref{thm:newmain}.}
Suppose $H$ is a connected graph in which each equivalence class of the degree partition contains at most two vertices. The {\sc $H$-Cover} problem can be solved in constant or linear time if $H$ is a tree or a cycle or a path (possibly ending with semi-edges). Otherwise, consider the reduced graph $H$, reduced via the {\dar} of Definition~\ref{def:reduction}. It is important that $H$ also has at most two vertices in each equivalence class of its degree partition. If $H$ is a path or a cycle, then {\sc $H$-Cover} is solvable in polynomial time, and so is {\sc $H$-Cover}, due to Observation~\ref{obs:reduction}.

If $H$ contains a vertex of degree greater than 2, then  all vertices of $H$ have degrees greater than 2.  If all monochromatic uniblock and interblock graphs of $H$ are harmless, then {\sc $H$-Cover} is polynomially solvable for general input graphs, and so is {\sc $H$-Cover}, due to Observation~\ref{obs:reduction}. If $H$ contains a harmful or a dangerous uniblock or interblock graph, then {\sc $H$-Cover} is NP-complete for simple input graphs by Parts 2 and 3 of Theorem~\ref{thm:main}. If $G$ is a simple graph as an input to the {\sc $H$-Cover} problem, the reverse operation to the {\dar} gives a simple graph $G$ such that $G\to H$ if and only if $G \to H$. Hence {\sc $H$-Cover} is also NP-complete for simple input graphs. This concludes the proof.
\qed

\section{Concluding remarks}\label{sec:concl}

The polynomial algorithm described in Section~\ref{sec:poly} combines two approaches: finding perfect matchings and solving 2-SAT. Both of these problems are well known to be solvable in polynomial time. However, it is somewhat surprising that their combination also remains polynomially solvable. This contrasts with the so-called compatible 2-factor problem~\cite{kratochvil1992compatible}, where instances that are solvable in polynomial time fall into two distinct categories --- one solved by a reduction to perfect matching and the other solved by 2-SAT. When restrictions from both categories are present in the same instance, the problem becomes NP-complete.

Moreover, the tractability of the polynomial-time solvable cases does not depend on the target graph being fixed. If $H$ is a graph with at most two vertices in each block of the degree partition, and all monochromatic block and interblock graphs are harmless, then the algorithm described in Section~\ref{sec:poly} remains polynomial even when $H$ is part of the input.

We believe that the method developed above has a much wider potential and we conjecture the following:

\begin{conjecture2}
Let $H$ be a  block graph of a graph $H'$. Then {\sc $H$-Cover} for simple input graphs polynomially reduces to {\sc $H'$-Cover} for simple input graphs.
\end{conjecture2}

Ultimately, the overarching goal is to prove (or disprove) the Strong Dichotomy Conjecture for graph covers parameterized by the target graph. Ideally, this would include a comprehensive catalogue of all polynomially solvable cases.

\bibliographystyle{plainurl}
\bibliography{0-main,bib/knizky,bib/nakryti,bib/sborniky}

\begin{thebibliography}{10}

\bibitem{n:AFS91}
James Abello, Michael~R. Fellows, and John~C. Stillwell.
\newblock On the complexity and combinatorics of covering finite complexes.
\newblock {\em Australasian Journal of Combinatorics}, 4:103--112, 1991.
\newblock URL: \url{https://ajc.maths.uq.edu.au/pdf/4/ocr-ajc-v4-p103.pdf}.

\bibitem{n:Angluin80}
Dana Angluin.
\newblock Local and global properties in networks of processors.
\newblock {\em Proceedings of the 12th ACM Symposium on Theory of Computing},
  pages 82--93, 1980.
\newblock \href {https://doi.org/10.1145/800141.804655}
  {\path{doi:10.1145/800141.804655}}.

\bibitem{n:Biggs74}
Norman Biggs.
\newblock {\em Algebraic Graph Theory}.
\newblock Cambridge University Press, 1974.

\bibitem{n:Biggs81}
Norman Biggs.
\newblock Covering biplanes.
\newblock In {\em The theory and applications of graphs, Fourth International
  Conference, Kalamazoo}, pages 73--79. John Wiley \& Sons., 1981.

\bibitem{n:Biggs82}
Norman Biggs.
\newblock Constructing 5-arc transitive cubic graphs.
\newblock {\em Journal of London Mathematical Society II.}, 26:193--200, 1982.
\newblock \href {https://doi.org/10.1112/jlms/s2-26.2.193}
  {\path{doi:10.1112/jlms/s2-26.2.193}}.

\bibitem{n:Biggs84}
Norman Biggs.
\newblock Homological coverings of graphs.
\newblock {\em Journal of London Mathematical Society II.}, 30:1--14, 1984.
\newblock \href {https://doi.org/10.1112/jlms/s2-30.1.1}
  {\path{doi:10.1112/jlms/s2-30.1.1}}.

\bibitem{n:Bodlaender89}
Hans~L. Bodlaender.
\newblock The classification of coverings of processor networks.
\newblock {\em Journal of Parallel Distributed Computing}, 6:166--182, 1989.
\newblock \href {https://doi.org/10.1016/0743-7315(89)90048-8}
  {\path{doi:10.1016/0743-7315(89)90048-8}}.

\bibitem{n:BFJKR24-Algorithmica}
Jan Bok, Jir{\'{\i}} Fiala, Nikola Jedlickov{\'{a}}, Jan Kratochv{\'{\i}}l, and
  Pawel Rzazewski.
\newblock List covering of regular multigraphs with semi-edges.
\newblock {\em Algorithmica}, 86(3):782--807, 2024.
\newblock URL: \url{https://doi.org/10.1007/s00453-023-01163-7}, \href
  {https://doi.org/10.1007/S00453-023-01163-7}
  {\path{doi:10.1007/S00453-023-01163-7}}.

\bibitem{n:BFJKS23-WG}
Jan Bok, Jir{\'{\i}} Fiala, Nikola Jedlickov{\'{a}}, Jan Kratochv{\'{\i}}l, and
  Michaela Seifrtov{\'{a}}.
\newblock Computational complexity of covering colored mixed multigraphs with
  degree partition equivalence classes of size at most two (extended abstract).
\newblock In Dani{\"{e}}l Paulusma and Bernard Ries, editors, {\em
  Graph-Theoretic Concepts in Computer Science - 49th International Workshop,
  {WG} 2023}, volume 14093 of {\em Lecture Notes in Computer Science}, pages
  101--115. Springer, 2023.
\newblock \href {https://doi.org/10.1007/978-3-031-43380-1\_8}
  {\path{doi:10.1007/978-3-031-43380-1\_8}}.

\bibitem{n:BFHJK21-MFCS}
Jan Bok, Jiří Fiala, Petr Hliněný, Nikola Jedličková, and Jan
  Kratochvíl.
\newblock Computational complexity of covering multigraphs with semi-edges:
  Small cases.
\newblock In Filippo Bonchi and Simon~J. Puglisi, editors, {\em 46th
  International Symposium on Mathematical Foundations of Computer Science},
  volume 202 of {\em LIPIcs}, pages 21:1--21:15. Schloss Dagstuhl ---
  Leibniz-Zentrum für Informatik, 2021.
\newblock \href {https://doi.org/10.4230/LIPIcs.MFCS.2021.21}
  {\path{doi:10.4230/LIPIcs.MFCS.2021.21}}.

\bibitem{n:BFJKS24-DAM}
Jan Bok, Jiří Fiala, Nikola Jedličková, Jan Kratochvíl, and Michaela
  Seifrtová.
\newblock Computational complexity of covering disconnected multigraphs.
\newblock {\em Discrete Applied Mathematics}, 359:229--243, 2024.
\newblock URL: \url{https://doi.org/10.1016/j.dam.2024.07.035}, \href
  {https://doi.org/10.1016/J.DAM.2024.07.035}
  {\path{doi:10.1016/J.DAM.2024.07.035}}.

\bibitem{n:BilkaJKTV11}
Ondřej Bílka, Jozef Jirásek, Pavel Klavík, Martin Tancer, and Jan Volec.
\newblock On the complexity of planar covering of small graphs.
\newblock In Petr Kolman and Jan Kratochvíl, editors, {\em Graph-Theoretic
  Concepts in Computer Science}, volume 6986 of {\em Lecture Notes in Computer
  Science}, pages 83--94. Springer, 2011.
\newblock \href {https://doi.org/10.1007/978-3-642-25870-1_9}
  {\path{doi:10.1007/978-3-642-25870-1_9}}.

\bibitem{n:BLT11}
Ondřej Bílka, Bernard Lidický, and Marek Tesař.
\newblock Locally injective homomorphism to the simple weight graphs.
\newblock In Mitsunori Ogihara and Jun Tarui, editors, {\em Theory and
  Applications of Models of Computation}, volume 6648 of {\em Lecture Notes in
  Computer Science}, pages 471--482. Springer, 2011.
\newblock \href {https://doi.org/10.1007/978-3-642-20877-5_46}
  {\path{doi:10.1007/978-3-642-20877-5_46}}.

\bibitem{n:Chalopin05}
Jérémie Chalopin.
\newblock Local computations on closed unlabelled edges: the election problem
  and the naming problem.
\newblock In Peter Vojtáš, Mária Bieliková, Bernadette Charron-Bost, and
  Ondřej Sýkora, editors, {\em SOFSEM 2005: Theory and Practice of Computer
  Science}, volume 3381 of {\em Lecture Notes in Computer Science}, pages
  82--91. Springer, 2005.
\newblock \href {https://doi.org/10.1007/978-3-540-30577-4_11}
  {\path{doi:10.1007/978-3-540-30577-4_11}}.

\bibitem{n:ChMZ06}
Jérémie Chalopin, Yves Métivier, and Wiesław Zielonka.
\newblock Local computations in graphs: the case of cellular edge local
  computations.
\newblock {\em Fundamenta Informaticae}, 74(1):85--114, 2006.
\newblock \href {https://doi.org/10.5555/1231199.1231204}
  {\path{doi:10.5555/1231199.1231204}}.

\bibitem{n:ChalopinP11}
Jérémie Chalopin and Daniël Paulusma.
\newblock Graph labelings derived from models in distributed computing: {A}
  complete complexity classification.
\newblock {\em Networks}, 58(3):207--231, 2011.
\newblock \href {https://doi.org/10.1002/net.20432}
  {\path{doi:10.1002/net.20432}}.

\bibitem{n:CorneilG70}
Derek~G. Corneil and Calvin~C. Gotlieb.
\newblock An efficient algorithm for graph isomorphism.
\newblock {\em Journal of the Association for Computing Machinery}, 17:51--64,
  1970.
\newblock \href {https://doi.org/10.1145/321556.321562}
  {\path{doi:10.1145/321556.321562}}.

\bibitem{n:CM94}
Bruno Courcelle and Yves Métivier.
\newblock Coverings and minors: {A}pplications to local computations in graphs.
\newblock {\em European Journal of Combinatorics}, 15:127--138, 1994.
\newblock \href {https://doi.org/10.1006/eujc.1994.1015}
  {\path{doi:10.1006/eujc.1994.1015}}.

\bibitem{n:FKKN14}
Jiří Fiala, Pavel Klavík, Jan Kratochvíl, and Roman Nedela.
\newblock Algorithmic aspects of regular graph covers with applications to
  planar graphs.
\newblock {\em CoRR}, abs/1402.3774, 2014.

\bibitem{n:FK01}
Jiří Fiala and Jan Kratochvíl.
\newblock Complexity of partial covers of graphs.
\newblock In Peter Eades and Tadao Takaoka, editors, {\em Algorithms and
  Computation}, volume 2223 of {\em Lecture Notes in Computer Science}, pages
  537--549. Springer, 2001.
\newblock \href {https://doi.org/10.1007/3-540-45678-3_46}
  {\path{doi:10.1007/3-540-45678-3_46}}.

\bibitem{n:FK08}
Jiří Fiala and Jan Kratochvíl.
\newblock Locally constrained graph homomorphisms --- structure, complexity,
  and applications.
\newblock {\em Computer Science Review}, 2(2):97--111, 2008.
\newblock \href {https://doi.org/10.1016/J.COSREV.2008.06.001}
  {\path{doi:10.1016/J.COSREV.2008.06.001}}.

\bibitem{n:FKP08}
Jiří Fiala, Jan Kratochvíl, and Attila Pór.
\newblock On the computational complexity of partial covers of {T}heta graphs.
\newblock {\em Discrete Applied Mathematics}, 156:1143--1149, 2008.
\newblock \href {https://doi.org/10.1016/J.DAM.2007.05.051}
  {\path{doi:10.1016/J.DAM.2007.05.051}}.

\bibitem{FS25}
Jiří Fiala and Michaela Seifrtová.
\newblock A novel approach to covers of multigraphs with semi-edges.
\newblock {\em Discussiones Mathematicae Graph Theory}, accepted, in press.
\newblock \href {https://doi.org/10.7151/dmgt.2540}
  {\path{doi:10.7151/dmgt.2540}}.

\bibitem{getzler1998modular}
Ezra Getzler and Mikhail~M Kapranov.
\newblock Modular operads.
\newblock {\em Compositio Mathematica}, 110(1):65--125, 1998.
\newblock \href {https://doi.org/10.1023/A:1000245600345}
  {\path{doi:10.1023/A:1000245600345}}.

\bibitem{k:GT87}
Jonathan~L. Gross and Thomas~W. Tucker.
\newblock {\em Topological Graph Theory}.
\newblock J. Wiley and Sons, 1987.

\bibitem{kratochvil2003complexity}
Jan Kratochvíl.
\newblock Complexity of hypergraph coloring and {S}eidel's switching.
\newblock In Hans~L. Bodlaender, editor, {\em Graph-Theoretic Concepts in
  Computer Science}, volume 2880 of {\em Lecture Notes in Computer Science},
  pages 297--308. Springer, 2003.
\newblock \href {https://doi.org/10.1007/978-3-540-39890-5_26}
  {\path{doi:10.1007/978-3-540-39890-5_26}}.

\bibitem{kratochvil1992compatible}
Jan Kratochvíl and Svatopluk Poljak.
\newblock Compatible 2-factors.
\newblock {\em Discrete Applied Mathematics}, 36(3):253--266, 1992.
\newblock \href {https://doi.org/10.1016/0166-218X(92)90257-B}
  {\path{doi:10.1016/0166-218X(92)90257-B}}.

\bibitem{n:KPT94}
Jan Kratochvíl, Andrzej Proskurowski, and Jan~Arne Telle.
\newblock Complexity of graph covering problems.
\newblock In Ernst~W. Mayr, Gunther Schmidt, and Gottfried Tinhofer, editors,
  {\em Graph-Theoretic Concepts in Computer Science}, volume 903 of {\em
  Lecture Notes in Computer Science}, pages 93--105. Springer, 1994.
\newblock \href {https://doi.org/10.1007/3-540-59071-4_40}
  {\path{doi:10.1007/3-540-59071-4_40}}.

\bibitem{n:KPT97a}
Jan Kratochvíl, Andrzej Proskurowski, and Jan~Arne Telle.
\newblock Covering directed multigraphs {I.} {C}olored directed multigraphs.
\newblock In Rolf~H. Möhring, editor, {\em Graph-Theoretic Concepts in
  Computer Science}, volume 1335 of {\em Lecture Notes in Computer Science},
  pages 242--257. Springer, 1997.
\newblock \href {https://doi.org/10.1007/BFB0024502}
  {\path{doi:10.1007/BFB0024502}}.

\bibitem{n:KPT97}
Jan Kratochvíl, Andrzej Proskurowski, and Jan~Arne Telle.
\newblock Covering regular graphs.
\newblock {\em Journal of Combinatorial Theory, Series B}, 71(1):1--16, 1997.
\newblock \href {https://doi.org/10.1006/JCTB.1996.1743}
  {\path{doi:10.1006/JCTB.1996.1743}}.

\bibitem{n:KTT16}
Jan Kratochvíl, Jan~Arne Telle, and Marek Tesař.
\newblock Computational complexity of covering three-vertex multigraphs.
\newblock {\em Theoretical Computer Science}, 609:104--117, 2016.
\newblock \href {https://doi.org/10.1016/j.tcs.2015.09.013}
  {\path{doi:10.1016/j.tcs.2015.09.013}}.

\bibitem{kwak2007graphs}
Jin~Ho Kwak and Roman Nedela.
\newblock Graphs and their coverings.
\newblock {\em Lecture Notes Series}, 17, 2007.

\bibitem{n:LMZ93}
Igor Litovsky, Yves Métivier, and Wiesław Zielonka.
\newblock The power and the limitations of local computations on graphs.
\newblock In Ernst~W. Mayr, editor, {\em Graph-Theoretic Concepts in Computer
  Science}, volume 657 of {\em Lecture Notes in Computer Science}, pages
  333--345. Springer, 1992.
\newblock \href {https://doi.org/10.1007/3-540-56402-0_58}
  {\path{doi:10.1007/3-540-56402-0_58}}.

\bibitem{n:MalnivcMP04}
Aleksander Malnič, Dragan Marušič, and Primož Potočnik.
\newblock Elementary abelian covers of graphs.
\newblock {\em Journal of Algebraic Combinatorics}, 20(1):71--97, 2004.
\newblock \href {https://doi.org/10.1023/B:JACO.0000047294.42633.25}
  {\path{doi:10.1023/B:JACO.0000047294.42633.25}}.

\bibitem{n:MalnicNS00}
Aleksander Malnič, Roman Nedela, and Martin Škoviera.
\newblock Lifting graph automorphisms by voltage assignments.
\newblock {\em European Journal of Combinatorics}, 21(7):927--947, 2000.
\newblock \href {https://doi.org/10.1006/eujc.2000.0390}
  {\path{doi:10.1006/eujc.2000.0390}}.

\bibitem{nedela_mednykh}
Alexander~D. Mednykh and Roman Nedela.
\newblock {\em Harmonic Morphisms of Graphs: Part I: Graph Coverings}.
\newblock Vydavatelstvo Univerzity Mateja Bela v Banskej Bystrici, 1st edition,
  2015.

\bibitem{n:NedelaS96}
Roman Nedela and Martin Škoviera.
\newblock Regular embeddings of canonical double coverings of graphs.
\newblock {\em Journal of Combinatorial Theory, Series B}, 67(2):249--277,
  1996.
\newblock \href {https://doi.org/10.1006/jctb.1996.0044}
  {\path{doi:10.1006/jctb.1996.0044}}.

\bibitem{k:Reidemeister32}
Kurt Reidemeister.
\newblock {\em {Einf\"uhrung in die kombinatorische Topologie.}}
\newblock {Braunschweig: Friedr. Vieweg\&Sohn A.-G. XII, 209 S. }, 1932.

\end{thebibliography}

\end{document}